\documentclass[11pt,twoside]{article}
\usepackage{fancyhdr}
\usepackage[colorlinks,citecolor=blue,urlcolor=blue,linkcolor=blue,bookmarks=false]{hyperref}
\usepackage{amsfonts,epsfig,graphicx}
\usepackage{afterpage}
\usepackage{comment}
\usepackage{amsmath,amssymb,amsthm} 
\usepackage{mathtools}
\usepackage{enumerate}
\usepackage{fullpage}
\usepackage[T1]{fontenc} 
\usepackage{epsf} 
\usepackage{graphics} 
\usepackage{amsfonts,amsmath}
\usepackage{natbib} 
\usepackage{psfrag,xspace}
\usepackage{color,etoolbox}
\usepackage{subcaption} 
\usepackage{dsfont}
\usepackage[page]{appendix}
\usepackage{soul} 

\usepackage{enumitem}

\DeclareSymbolFont{bbold}{U}{bbold}{m}{n}
\DeclareSymbolFontAlphabet{\mathbbold}{bbold}

\newtheorem{theorem}{Theorem}
\newtheorem{lemma}{Lemma}
\newtheorem{corollary}{Corollary}
\newtheorem{proposition}{Proposition}

\newtheorem{assumption}{Assumption}
\newtheorem{remark}{Remark}
\theoremstyle{definition}

\theoremstyle{remark}

\newcommand\independent{\protect\mathpalette{\protect\independenT}{\perp}}
\def\independenT#1#2{\mathrel{\rlap{$#1#2$}\mkern2mu{#1#2}}}

\usepackage{pgf,tikz}
\usetikzlibrary{positioning}
\usetikzlibrary{shapes}
\usepackage{pgfplots}
\pgfplotsset{compat=1.18}
\usetikzlibrary{matrix}       
\usetikzlibrary{arrows,shapes,positioning,shadows,trees}

\begin{document}

\def\spacingset#1{\renewcommand{\baselinestretch}%
{#1}\small\normalsize} \spacingset{1}

\raggedbottom
\allowdisplaybreaks[1]


  \title{\vspace*{-.4in} {Nonparametric identification and efficient estimation of causal effects with instrumental variables}}
  \author{Alexander W. Levis$^1$, Edward H. Kennedy$^1$, Luke Keele$^2$ \\  \\ \\
    $^1$Department of Statistics \& Data Science, \\
    Carnegie Mellon University \\
    $^2$Department of Surgery and Biostatistics, \\
    University of Pennsylvania \\ \\ 
    \texttt{alevis@cmu.edu}; \\ \texttt{edward@stat.cmu.edu}; \\
    \texttt{luke.keele@uphs.upenn.edu}
\date{}
    }
    
  \maketitle
  \thispagestyle{empty}

\begin{abstract}
Instrumental variables are widely used in econometrics and epidemiology for identifying and estimating causal effects when an exposure of interest is confounded by unmeasured factors. Despite this popularity, the assumptions invoked to justify the use of instruments differ substantially across the literature. Similarly, statistical approaches for estimating the resulting causal quantities vary considerably, and often rely on strong parametric assumptions. In this work, we compile and organize structural conditions that nonparametrically identify conditional average treatment effects, average treatment effects among the treated, and local average treatment effects, with a focus on identification formulae invoking the conditional Wald estimand. Moreover, we build upon existing work and propose nonparametric efficient estimators of functionals corresponding to marginal and conditional causal contrasts resulting from the various identification paradigms. We illustrate the proposed methods on an observational study examining the effects of operative care on adverse events for cholecystitis patients, and a randomized trial assessing the effects of market participation on political views.
\end{abstract}

\bigskip


\pagebreak

\section{Introduction}
In studies of causal effects, hidden confounding---the presence of unobserved factors that relate to both the exposure and outcome of interest---poses a critical challenge. Assigning treatments via a randomized mechanism is one method based on the study design to avoid unobserved confounders and mitigate bias. Another approach is to assume that selection into treatment groups is based on observed data only; that is, to assume there are no unobserved confounders. A popular alternative to the assumption of no unmeasured confounding is to make use of an instrumental variable. An instrument is a variable that affects the exposure of interest but does not directly affect the outcome, and is itself unconfounded \citep{Angrist:1996,Hernan:2006}. If $Z$ denotes the instrument, $A$ denotes the treatment actually received, and $Y$ denotes the outcome of interest, then $Z$ can be used to estimate an effect of $A$ on $Y$ even in the presence of unobserved confounding.

An instrument variable (IV) can arise in two main ways. First, in randomized experiments with noncompliance, randomization to treatment arms serves as an instrument for exposure to the experimental protocol. Second, researchers can identify natural circumstances that provide haphazard encouragement for units to be exposed to a treatment. Nowadays, methodologies based on instrumental variables see widespread use across a number of disciplines of study including economics \citep{Angrist:2001,imbens2014instrumental}, political science \citep{Keele:2014b}, and epidemiology \citep{Hernan:2006, Baiocchi:2014}, including the Mendelian randomization paradigm used in studies in genetic epidemiology \citep{davey_smith_mendelian_2003, davey_smith_mendelian_2004, lawlor_mendelian_2008}. 

Ideally, an IV analysis should follow what has become the template for modern causal analysis.  The first step of this template is to define the target causal effect. Often this consists of defining the causal effect or estimand as a counterfactual contrast. Next, the researcher should state or derive the assumptions needed to identify the estimand from observed data---that is, represent the causal estimand as a statistical functional. Given that many causal analyses depend on assumptions that are untestable, researchers should also provide substantive justification for the plausibility of these assumptions. The final step is estimation of the causal effects from the data. In many applications, there is an interplay between the choice of the causal estimand and the assumptions needed for identification. That is, analysts can often select from among more than one causal estimand, and each estimand may require different assumptions. As such, the degree of plausibility of the identification assumptions may be a direct consequence of the choice of the estimand. This interplay between causal estimands and assumptions is particularly acute in an instrumental variable analysis.

In an IV analysis, applied investigators may wish to target the average treatment effect (ATE), but identification of the ATE typically relies on effect homogeneity or no-interaction assumptions \citep{robins1994b, Hernan:2006, tan2010, wang2018}.  Alternatively, investigators may instead focus on an estimand known as the local average treatment effect (LATE), which is the average causal effect among those who comply with encouragement to treatment \citep{Angrist:1996}. Identification of the LATE estimand requires a monotonicity assumption, which may be plausible in many settings, but has more limited interpretability and applicability than the ATE. A variety of other estimands---including the average treatment effect among the treated (ATT)---can be targeted in an IV analysis, each with assumptions that vary in plausibility depending on the context.

In this paper, we focus on two related topics. The vast majority of IV studies rely on a set of three core assumptions, but these assumptions do not allow for point identification of any target estimand. As such, additional assumptions are needed for point identification, and this class of additional assumptions dictate the causal quantity that is identified in a specific IV study. First, we provide a review of additive causal contrasts that can be point identified in an IV analysis. For each of these estimands, we review the key assumptions needed for point identification. Additionally, for each estimand we review empirical applications where those assumptions are plausible or implausible. Second, we outline a unified estimation framework that encompasses each of the estimands we reviewed. More specifically, we invoke a framework that is often referred to as doubly-robust machine learning (DRML) that combines semiparametrics, doubly robust (DR) methods, nonparametrics, and machine learning~\citep{bickel1993efficient,van2003unified,tsiatis2006semiparametric, kennedy2016semiparametric, kennedy2022}. In this framework, DR methods are combined with ML estimation and sample-splitting to achieve fast parameteric-type convergence rates with flexible estimation, i.e., less bias while preserving efficiency. While other general overviews of the IV framework exist in the literature~\citep{Baiocchi:2014, clarke2012, imbens2014instrumental}, our review is the first to focus on the variety of possible estimands and the related assumptions, and provide a unified estimation framework based on DRML methods.

\subsection{Other Related Topics}
As mentioned above, our main focus in this paper is on identifying assumptions that underlie additive causal contrasts, and subsequent efficient estimation, in settings with an IV. We emphasize that this is by no means comprehensive of the set of topics that have been explored in the IV literature. Indeed, research on IV methods, which was once mostly confined to econometrics, has become an active area of research in statistics and has expanded rapidly over the last twenty years. Much of this literature is focused on improving methods of IV estimation and inference. For example, one strand of the literature that we do not focus on in this paper has considered IV-based estimation of causal effects on other scales, e.g., risk ratios and odds ratios \citep{Vansteelandt:2003,robins2004,Clarke:2010,Clarke:2012}. Other work has focused on estimation when there are multiple IVs \citep{kang2016instrumental,burgess2017review,windmeijer2019use}. In particular, estimation and inference with weak instruments (i.e., settings where one or more IVs are only weakly correlated with the exposure of interest), which pose significant challenges for estimation and inference, has been an active area of development \citep{andrews2007performance,Keele:2017fiv, andrews2019weak}. This research includes the development of methods that effectively strengthen instruments that are weak via matching or other methods \citep{Baiocchi:2010,Zubizarreta:2012a,heng2019instrumental,keeleicumatch2019,kennedy2020b}. Finally, moving beyond point identification, there has been considerable work in developing methods for partial identification of causal estimands in IV settings. This research includes methods for sensitivity analysis \citep{Rosenbaum:1996,Baiocchi:2014,fogarty2021biased}, as well as the characterization of nonparametric bounds \citep{Balke:1997,heckman1999local,manski2000monotone,swanson2018partial,kennedy2018survivor, levis2023}. 

\section{Notation and Core Assumptions}\label{sec:notation}


First, we outline our notation. We use \(O = (\boldsymbol{X}, Z, A, Y) \sim P\) to denote the collection of observed quantities. Specifically, \(Y \in \mathcal{Y} \subseteq \mathbb{R}\) is the outcome, \(A \in \mathcal{A} \subseteq \mathbb{R}\) is the treatment or exposure of interest, \(Z \in \mathcal{Z} \subseteq \mathbb{R}\) is the instrumental variable, and \(\boldsymbol{X} \in \mathcal{X} \subseteq \mathbb{R}^d\) is the set of observed, baseline covariates. Frequently, we will focus on the case of a binary instrument, $\mathcal{Z} = \{0,1\}$, and/or binary exposure, $\mathcal{A} = \{0,1\}$. For applications with no measured covariates, we set \(\boldsymbol{X} = \emptyset\). 

Throughout, we use potential outcomes to denote counterfactual quantities. We use \(Y(a)\) to denote the counterfactual outcome that \emph{would have been observed} had treatment \(A\) been set to \(a \in \mathcal{A}\). We write \(A(z)\), \(Y(z)\) for the counterfactual treatment and outcome for \(A\) and \(Y\) when setting \(Z\) to \(z \in \mathcal{Z}\). Finally, we use \(Y(z,a)\) to denote the counterfactual outcome when setting \((Z, A)\) to \((z, a) \in \mathcal{Z} \times \mathcal{A}\). Critically, our notation encodes two important causal assumptions: first, that there are no hidden forms of the instrument or treatment, and second, that a subject's potential outcome is not affected by other subjects' instrument or exposure values. Together, these allow us to link counterfactual and observed data, such that for those subjects with $Z=z$ and $A=a$, we have $Y(z, a) = Y$, which is often referred to as the consistency assumption in the epidemiology literature. These two assumptions together are often referred to as the stable unit treatment value assumption (SUTVA)~\citep{Rubin:1986}. We assume consistency for the remaining two counterfactual quantities as well, i.e., we assume $A(Z) = A$ and $Y(Z) = Y$. These assumptions are not unique to IV methods, and are invoked in a variety of causal inference settings.

The primary goal is to use IV methods to identify and estimate the \(A\rightarrow Y\) treatment effect, e.g., the mean of \(Y(a) - Y(a')\) for \(a, a' \in \mathcal{A}\), possibly among some subset of the population. When the treatment-outcome relationship is confounded (e.g., when $A$ is not randomized), estimating the effect of $A$ on $Y$ is difficult. However, if a valid instrumental variable, $Z$, is available, a causal effect of $A$ on $Y$ can be estimated consistently---under critical assumptions to be discussed---even in the presence of unobserved confounders $U$ between $Y$ and $A$. Informally, $Z$ is a valid IV if it 1) affects or is associated with $A$, 2) is as good as randomized, possibly after conditioning on measured covariates $\boldsymbol{X}$, and 3) affects the outcome $Y$ only indirectly through $A$. These three ``core'' IV assumptions are typically written formally as follows.
\begin{assumption}[Relevance]
\label{ass:relevance}
    $\mathrm{Cov}_P(Z, A \mid \boldsymbol{X}) \neq 0$, almost surely.
\end{assumption}

\begin{assumption}[Unconfoundedness]
\label{ass:UC}
\
    \begin{enumerate}[label=(\alph*),leftmargin=.5in]
        \item $Z \independent Y(z, a) \mid \boldsymbol{X}$, for all $a \in \mathcal{A}$
        \item $Z \independent (A(z), Y(z)) \mid \boldsymbol{X}$, for all $z \in \mathcal{Z}$
        \item $(Z, A) \independent Y(z, a) \mid \boldsymbol{X}, U$, for all $z \in\mathcal{Z}$, $a \in \mathcal{A}$
        \item $Z \independent U \mid \boldsymbol{X}$
    \end{enumerate}
\end{assumption}

\begin{assumption}[Exclusion Restriction]
\label{ass:ER}
    $Y(z,a) = Y(a)$ for any $z \in \mathcal{Z}$ and $a \in \mathcal{A}$.
\end{assumption}

\begin{remark}
    When the instrument is binary (i.e., $\mathcal{Z} = \{0,1\}$), we often require and invoke the positivity condition,
\begin{equation}\label{eq:pos}
      0< P[Z = 1 \mid \boldsymbol{X}] < 1, \text{ almost surely}.
    \end{equation}
     We note that~\eqref{eq:pos} holds as a consequence of our relevance Assumption~\ref{ass:relevance}. To see this, note that if $P[P[Z = z \mid \boldsymbol{X}] = 1] > 0$, for some $z \in \{0,1\}$, then for the set $B = \{\boldsymbol{x} \in \mathcal{X}: P[Z = z \mid \boldsymbol{X} = \boldsymbol{x}] = 1\}$, we have $P(\mathrm{Cov}_P(Z, A \mid \boldsymbol{X}) = 0) \geq P(\boldsymbol{X} \in B) > 0$, since $Z$ is a constant on $B$, thus contradicting Assumption~\ref{ass:relevance}.
\end{remark}

\begin{remark}\label{remark:UC}
    While some authors prefer to work with a subset of the components of Assumption~\ref{ass:UC} (e.g., using only~\ref{ass:UC}(b), or~\ref{ass:UC}(c)\&(d)), we deliberately list all four conditions for this review. In each of the identification results in Section~\ref{sec:identification}, one or more components of Assumption~\ref{ass:UC} will be invoked, and we attempt to use only the weakest statement(s). Note, in particular, that by the properties of conditional independence, Assumptions~\ref{ass:UC}(c) and \ref{ass:UC}(d) imply that $Z \independent (U, Y(z, a)) \mid \boldsymbol{X}$, which implies Assumption~\ref{ass:UC}(a), i.e., (a) is weaker than (c)\&(d). Thus, the results in Section~\ref{sec:identification} that only depend on Assumption~\ref{ass:UC}(a) remain valid under a potentially wider set of scenarios than those that require Assumptions~\ref{ass:UC}(c)\&(d). In the Appendices, we also include a scenario---alternative to the mechanism illustrated in Figure~\ref{fig:IV}---in which Assumption~\ref{ass:UC}(a) holds, yet \ref{ass:UC}(b) fails.
\end{remark}


We now briefly discuss Assumptions~\ref{ass:relevance}--\ref{ass:ER} in turn. Assumption~\ref{ass:relevance} implies that there is a (linear) association between $Z$ and $A$, even after controlling for $\boldsymbol{X}$. That is, for $Z$ to be an instrument it must have some effect on treatment exposure, $A$. 
Assumption~\ref{ass:UC} formally expresses the fact that the IV must be unconfounded. Assumptions~\ref{ass:UC}(a) and~\ref{ass:UC}(b) can be interpreted as saying that the effect of $Z$ on $A$ and $Y$ is unconfounded after controlling for covariates $\boldsymbol{X}$. Assumption~\ref{ass:UC}(c) implies that the measured confounders $\boldsymbol{X}$ and unmeasured confounders $U$ would be sufficient to control for confounding of the effects of $Z$ and $A$ on $Y$. Assumption~\ref{ass:UC}(d) requires that the unmeasured confounders $U$ of the $A$-$Y$ relationship do not confound the effect of $Z$ on $Y$. Note that in the special case that $Z$ is marginally randomized (e.g., in a trial), Assumption~\ref{ass:UC} holds by design. In that case, Assumption~\ref{ass:UC} holds unconditionally, that is without conditioning on $\boldsymbol{X}$. Note, though, that conditioning on baseline covariates $\boldsymbol{X}$ in such an experiment will not affect the validity of Assumption~\ref{ass:UC}. Finally, Assumption~\ref{ass:ER} means that $Y(z,a)$ does not depend on $z$, i.e., any effect of the instrument on the outcome acts entirely through the exposure $A$.


While Assumption~\ref{ass:relevance} can be directly tested, Assumptions~\ref{ass:UC} and \ref{ass:ER} are both untestable with observed data. Justification of these two assumptions requires practitioners to appeal to qualitative reasoning. Alternatively, a variety of falsification tests have been proposed to probe the plausibility of Assumption~\ref{ass:UC}~\citep{jackson2015toward,davies2015commentary,davies2017compare,zhao2018,branson2020evaluating}. Assumption~\ref{ass:ER} can also be subject to falsification testing~\citep{glymour2012credible,kang2013causal,Yang:2013,pizer2016falsification,keeleexrest2018}. Throughout, we will assume these core assumptions hold. Other reviews provide in-depth discussions of these assumptions~\citep{hernan2019,Baiocchi:2014,Imbens:2015}.

Next, we introduce an application that we use to structure our discussion at various points. Here, we briefly discuss each of the core assumptions in the context of this application. REFLUX was a multicenter randomized controlled trial in the UK that randomized patients with Gastro-Oseophageal Reflux Disease (GORD) to either surgical management or Proton Pump Inhibitors (PPIs)~\citep{grant2008,grant2013}. In the trial, there was a substantial amount of non-adherence. Specifically, among the 178 participants in the surgical arm, 67 received medical management contrary to random assignment, and of 179 participants the medication arm, 10 crossed over and received surgery instead of medical management. Here, trial arm assignment serves as the instrument $Z$, and actual medical care received serves as the treatment $A$. The outcome $Y$ in the trial was a measure of quality of life. Note that in REFLUX, $\boldsymbol{X}$ includes the following covariates: age, sex, BMI, employment status (3 categories), age no longer in school, quality of life at baseline, and 5 scales that measured reflux symptoms including heartburn and nausea. For Assumption~\ref{ass:relevance} to hold, trial arm assignment must be related to medical care received: this is reasonable and was observed in the data. Next, since $Z$ was randomized, Assumption~\ref{ass:UC} holds by design, marginally or conditional on $\boldsymbol{X}$. Finally, Assumption~\ref{ass:ER} requires that $Z$ only has an effect on $Y$ through $A$.  That is, assignment to the surgery arm can only affect quality of life through exposure to surgery or medical management. Here, while the effect of $Z$ on $Y$ is straightforward to identify, we focus on the use of the IV framework to identify an effect of $A$ on $Y$, even in the presence of unobserved confounders between $A$ and $Y$.

The three core IV assumptions can also be encoded into a causal diagram \citep{pearl2009, richardson2013}. Figure~\ref{fig:IVa} contains a directed acyclic graph (DAG) for an IV design. Assumption~\ref{ass:relevance} is represented by the arrow from $Z$ to $A$. Assumptions~\ref{ass:UC}(a) and \ref{ass:UC}(b) hold given that $\boldsymbol{X}$ blocks all backdoor paths from $Z$ to $A$ and $Y$. Similarly, Assumptions~\ref{ass:UC}(c) and \ref{ass:UC}(d) hold as $(\boldsymbol{X}, U)$ blocks all backdoor paths from $Z$ and $A$ to $Y$, and $Z$ is $d$-separated from $U$ given $\boldsymbol{X}$. The case where $Z$ is (marginally) randomized is represented by the DAG in Figure~\ref{fig:IVb}, which omits the arrow from $\boldsymbol{X}$ to $Z$. In this scenario, Assumption~\ref{ass:UC} holds without conditioning on $\boldsymbol{X}$, but continues to hold conditional on $\boldsymbol{X}$. Here, $\boldsymbol{X}$ simply represents measured baseline confounders between $A$ and $Y$; this scenario would be the expected structure in the REFLUX example. Finally, whether in the DAG of Figure~\ref{fig:IVa} or~\ref{fig:IVb}, Assumption~\ref{ass:ER} holds due to the fact that there is no arrow directly from $Z$ to $Y$ or $U$.  That is, any effect $Z$ has on $Y$ flows through $A$. Notably, either DAG allows for unobserved confounders $U$ between $A$ and $Y$.

\begin{figure}
\centering




\begin{subfigure}{.5\textwidth}
\large{\begin{tikzpicture}[%
        ->,
        >=stealth,
        node distance=0.5cm,
        pil/.style={
          ->,
          thick,
          shorten =2pt,},
        regnode/.style={circle, draw=black, fill=none, thick, minimum size=8mm},
        bluenode/.style={circle, draw=blue, fill=none, thick, minimum size=8mm},
        boxednode/.style={rectangle, draw=black, fill=none, thick, minimum size=8mm},
        leftsplitnode/.style={semicircle, draw=black, fill=none, thick, minimum size=9mm, shape border rotate=90},
        rightsplitnode/.style={semicircle, draw=red, fill=none, thick, minimum size=9mm, shape border rotate=270},
        rednode/.style={circle, draw=red, fill=none, thick, minimum size=8mm},
        redbox/.style={rectangle, draw=red, fill=none, thick, minimum size=8mm},
        ]
        \node[regnode] (1) {$Z$};
        \node[regnode, right=of 1] (2) {$A$};
        \node[regnode, above right=of 2] (3) {$U$};
        \node[regnode, below right=of 3] (4) {$Y$};
        \node[regnode, left=of 1] (5) {$\boldsymbol{X}$};
        \draw [->] (1) to (2);
        \draw [->] (2) to (4);
        \draw [->] (3) to (2);
        \draw [->] (3) to (4);
        \draw [->] (5) to (1);
        \draw [->] (5) to [out=330, in=210] (2);
        \draw [<->] (5) to [out=45, in=155] (3);
        \draw [->] (5) to [out=315, in=225] (4);
\end{tikzpicture}}

\caption{$Z$ unconfounded given $\boldsymbol{X}$}
\label{fig:IVa}
\end{subfigure}%
\begin{subfigure}{.5\textwidth}
\large{\begin{tikzpicture}[%
        ->,
        >=stealth,
        node distance=0.5cm,
        pil/.style={
          ->,
          thick,
          shorten =2pt,},
        regnode/.style={circle, draw=black, fill=none, thick, minimum size=8mm},
        bluenode/.style={circle, draw=blue, fill=none, thick, minimum size=8mm},
        boxednode/.style={rectangle, draw=black, fill=none, thick, minimum size=8mm},
        leftsplitnode/.style={semicircle, draw=black, fill=none, thick, minimum size=9mm, shape border rotate=90},
        rightsplitnode/.style={semicircle, draw=red, fill=none, thick, minimum size=9mm, shape border rotate=270},
        rednode/.style={circle, draw=red, fill=none, thick, minimum size=8mm},
        redbox/.style={rectangle, draw=red, fill=none, thick, minimum size=8mm},
        ]
        \node[regnode] (1) {$Z$};
        \node[regnode, right=of 1] (2) {$A$};
        \node[regnode, above right=of 2] (3) {$U$};
        \node[regnode, below right=of 3] (4) {$Y$};
        \node[regnode, left=of 1] (5) {$\boldsymbol{X}$};
        \draw [->] (1) to (2);
        \draw [->] (2) to (4);
        \draw [->] (3) to (2);
        \draw [->] (3) to (4);
        \draw [->] (5) to [out=330, in=210] (2);
        \draw [<->] (5) to [out=45, in=155] (3);
        \draw [->] (5) to [out=315, in=225] (4);
\end{tikzpicture}}

\caption{$Z$ marginally randomized}
\label{fig:IVb}
\end{subfigure}
\caption{Causal diagram for the IV design}
\label{fig:IV}
\end{figure}
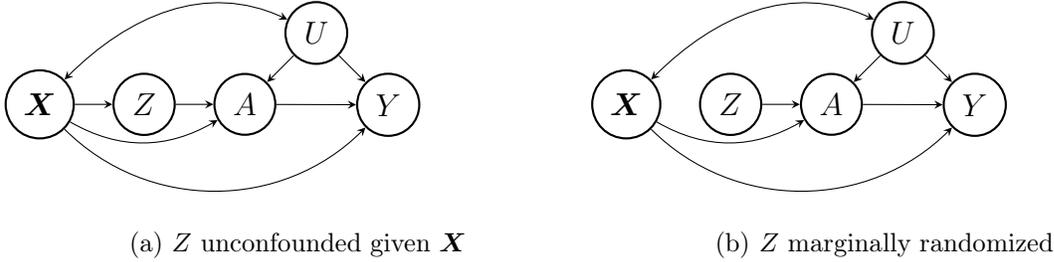

Finally, we outline the notation for a nonparametric structural equation model (NPSEM) that describes IV methods. We use this NPSEM to provide a unifying framework for the full set of possible IV assumptions. From the DAG in Figure \ref{fig:IV}, we can infer a general NPSEM, in which each variable is generated as a function of its parents, as well as some variable-specific noise~\citep{pearl2009, richardson2013}. In IV problems, the most critical structural equation---and the only one we will invoke in this paper---is that for the outcome. Concretely, we write the NPSEM for the outcome as
\begin{equation} \label{eq:structY}
Y = f_Y(\boldsymbol{X}, U, A, \epsilon_Y),
\end{equation}
where $f_Y$ is an unknown function and $\epsilon_Y$ is a noise variable. Since we view this as a \textit{structural} equation within the model implied by Figure~\ref{fig:IV}, how $Y$ behaves under specific interventions can be determined by propagating fixed values of intervened variables to all of their descendants. For example, intervening to set $A=a$, we represent this in the model as $Y(a) = f_Y(\boldsymbol{X}, U, a, \epsilon_Y)$. Note that structural equation \eqref{eq:structY} directly encodes the exclusion restriction, as $Z$ does not directly affect how $Y$ is generated. Written formally, we have $Y(z, a) = f_Y(\boldsymbol{X}, U, a, \epsilon_Y) = Y(a)$, since $\boldsymbol{X}$ and $U$ are not descendants of $Z$ in Figure \ref{fig:IV}. When $A$ is binary, we may write the NPSEM without any further assumptions as:
\begin{equation}\label{eq:general-binary}
    Y = h(\boldsymbol{X}, U, \epsilon_Y)A + g(\boldsymbol{X}, U, \epsilon_Y),
\end{equation}
for some unknown functions $g$ and $h$, which represent $Y(a = 0)$ and $\{Y(a = 1) - Y(a = 0)\}$, respectively. Since primary interest lies in the effect of $A$ on $Y$, many existing structural assumptions place restrictions on $h$. 

As we highlighted above, the primary purpose of the IV framework is to use the presence of $Z$ to consistently estimate an effect of $A$ on $Y$ even in the presence of the unobserved confounding by $U$ in Figure~\ref{fig:IV}. Critically, the three core assumptions are not in themselves sufficient for point identification of $A\rightarrow Y$ treatment effects in a nonparametric model. In fact, under the three core assumptions, the most we can do is bound treatment effects~\citep{manski1990, balke1997}, or test the sharp null hypothesis under which the effect is zero for all units \citep{Hernan:2006}. Indeed, additional ``structural'' restrictions on the counterfactual treatment or outcome distributions are typically asserted so as to make certain causal effects point identifiable---this is why we have yet to formally define any estimands. 
The roadmap we now provide is focused on laying out the different structural assumptions that can be used for point identification of additive causal constrasts, and defining the specific estimands that can be targeted in each case.

\section{Identification of Causal Contrasts} 
\label{sec:identification}


In this section, we review the various structural assumptions that have been proposed for point identification in the IV framework. In the literature, two primary types of assumptions have been articulated for point identification. The first relies on differing types of assumptions about effect heterogeneity. The second relies on an assumption known as monotonicity. Our discussion expands on a similar review in Chapter 16 of \citet{hernan2019}. Proofs of all results can be found in the Appendices.

\subsection{Homogeneity Assumptions}
\label{sec:homogeneity}

We begin by reviewing the various forms of effect homogeneity assumptions that are in the literature. Homogeneity assumptions place various restrictions on how the effects of $A$ on $Y$ (or $Z$ on $A$) vary from unit to unit in the study population. Below, we outline several different homogeneity assumptions. The key advantage of homogeneity assumptions is that they allow for point identification of an estimand that is generally thought to be highly relevant. Unfortunately, homogeneity assumptions are often implausible or difficult to verify in specific applications. For each homogeneity assumption, we derive the primary estimand and consider how to interpret the assumption within the context of REFLUX. 

\subsubsection{Constant Additive Effects}
\label{sec:strict-homogeneity}

The strongest homogeneity assumption we consider requires that the effect of $A$ on $Y$ is constant from unit to unit in the study population, within levels of covariates $\boldsymbol{X}$. Specifically, we can represent this assumption with the following structural equation:
\begin{equation} \label{eq:strict-homogeneity}
Y = r(\boldsymbol{X})A + g(\boldsymbol{X}, U, \epsilon_Y),
\end{equation}
which implies $Y(a) - Y(a') = (a - a')r(\boldsymbol{X})$. Notably, the unmeasured confounders $U$ do not modify the effect of $A$ on $Y$. For a binary treatment, equation~\eqref{eq:strict-homogeneity} restricts equation~\eqref{eq:general-binary} so that $h(\boldsymbol{X}, U, \epsilon_Y) = Y(a = 1) - Y(a = 0) \equiv r(\boldsymbol{X})$. Since $r(\boldsymbol{X})$ represents the individual treatment effect, it is also necesarily the conditional average treatment effect (CATE), 
\[
\mathbb{E}(Y(a = 1)- Y(a = 0) \mid \boldsymbol{X})
\]
which is the average difference in outcomes when all individuals in the study population are assigned to one exposure value versus when all individuals are assigned to another exposure value (e.g., treatment versus control) averaging within the stratum of $\boldsymbol{X}$. The CATE in REFLUX is the average level of quality of life when all patients are assigned to surgery versus when all patients are assigned to PPIs, conditional on $\boldsymbol{X}$.

In the following proposition, we show that the CATE, $r(\boldsymbol{X})$, is identified under the core IV assumptions and the constant effect model~\eqref{eq:strict-homogeneity}.
\begin{proposition} \label{prop:homogeneity-ident}
    Under the IV Assumptions~\ref{ass:relevance}, \ref{ass:UC}(a), and model \eqref{eq:strict-homogeneity}, the constant treatment effect $r(\boldsymbol{X})$ is identified via the conditional Wald estimand,
    \[\Psi_P(\boldsymbol{X}) \coloneqq \frac{\mathrm{Cov}_P(Z, Y \mid \boldsymbol{X})}{\mathrm{Cov}_P(Z, A \mid \boldsymbol{X})}.\]
\end{proposition}

First, we note that under the constant additive effect model, a simple consequence of Proposition~\ref{prop:homogeneity-ident} is that we can also point identify the following estimand:
\[
 \mathbb{E}(Y(a) - Y(a')) = (a - a')\mathbb{E}_P(\Psi_P(\boldsymbol{X})),
\]
an estimand commonly referred to as the average treatment effect (ATE). The ATE measures the unconditional average difference in outcomes. The ATE in REFLUX is the average level of quality of life when all patients are assigned to surgery versus when all patients are assigned to PPIs. Many analysts argue that the ATE is generally the most useful estimand from a policy perspective \citep{swanson2014think}. As we highlight in subsequent sections, other IV assumptions target estimands that differ from the CATE and ATE.

The function $\Psi_P(\boldsymbol{X})$ plays a prominent role in the IV framework. As we will see in subsequent sections, this function identifies causal effects across a variety of identification assumptions. The interpretation of this specific causal estimand, however, depends entirely on which assumptions are invoked. Importantly, since we will focus estimation efforts mostly on the case where $Z$ is binary, it helps to re-write the identification formula from Proposition \ref{prop:homogeneity-ident} for this setting. To that end, when $\mathcal{Z} = \{0,1\}$, we define the following nuisance functions: $\pi_z(\boldsymbol{X}) = P[Z = z \mid \boldsymbol{X}]$, $\mu_z(\boldsymbol{X}) = \mathbb{E}_P(Y \mid \boldsymbol{X}, Z = z)$, $\lambda_z(\boldsymbol{X}) = \mathbb{E}_P(A \mid \boldsymbol{X}, Z = z)$, for $z \in \{0,1\}$.

\begin{lemma}\label{lemma:binary}
    If $\mathcal{Z} = \{0,1\}$, then
    \[\Psi_P(\boldsymbol{X}) = \frac{\mu_1(\boldsymbol{X}) - \mu_0(\boldsymbol{X})}{\lambda_1(\boldsymbol{X}) - \lambda_0(\boldsymbol{X})}.\]
\end{lemma}

Inspection of $\Psi_P(\boldsymbol{X})$ in Lemma~\ref{lemma:binary} reveals that, for a binary instrument, we can interpret the numerator and denominator in the ratio as the conditional average treatment effects for the $Z\rightarrow Y$ and $Z \rightarrow A$ relationships, respectively, given $\boldsymbol{X}$.

In our view, model \eqref{eq:strict-homogeneity} is quite strong, and therefore unlikely to hold in real examples. Not only does this model require that unmeasured confounders not act as effect modifiers, but also that the treatment effect is \textit{exactly} the same across subjects in the same $\boldsymbol{X}$-stratum. Alternatively, the constant additive effect model with $\boldsymbol{X} = \emptyset$ asserts that the effect is exactly the same from unit to unit. In REFLUX, for this assumption to hold, the effect of surgery must be exactly the same from patient to patient across the whole population. More realistically, we might expect that some individuals who receive surgery may have a higher quality of life compared to treatment with PPIs, while quality of life is unaffected or is slightly worse for other patients that receive surgery. The remaining homogeneity assumptions place less restrictive constraints on how the effect of $A$ on $Y$ varies.

Before moving on, we note that traditional IV estimation methods based on two-stage least squares~\citep{angrist1995} can be justified under the constant additive effects model~\eqref{eq:strict-homogeneity}. Namely, suppose that \eqref{eq:strict-homogeneity} holds with $r(\boldsymbol{X}) = \boldsymbol{\beta}_0^T\phi(\boldsymbol{X})$ and $\mathbb{E}(g(\boldsymbol{X}, U, \epsilon_Y) \mid \boldsymbol{X}) = \boldsymbol{\alpha}_0^T \eta(\boldsymbol{X})$, for some known basis functions $\phi : \mathcal{X} \to \mathbb{R}^p$, $\eta: \mathcal{X} \to \mathbb{R}^s$, and unknown parameters $\boldsymbol{\beta}_0 \in \mathbb{R}^p$, $\boldsymbol{\alpha}_0 \in \mathbb{R}^s$.
If $\mathbb{E}(Y - \boldsymbol{\beta}_0^T \phi(\boldsymbol{X})A - \boldsymbol{\alpha}_0^T \eta(\boldsymbol{X})\mid \boldsymbol{X}, Z) = 0$, then the two-stage least squares procedure is a valid approach for estimating the CATE parameters $\boldsymbol{\beta}_0$: in the first stage, $\phi(\boldsymbol{X})A$ is regressed (linearly) on $q(\boldsymbol{X})Z$ and $\eta(\boldsymbol{X})$ for some user-specified function $q: \mathcal{X} \to \mathbb{R}^p$, and in the second stage, $Y$ is regressed on the fitted values from the first stage and $\eta(\boldsymbol{X})$; $\widehat{\boldsymbol{\beta}}$ are the estimated coefficients on the first-stage fitted values.


\subsubsection{No \texorpdfstring{$U$-$A$}{U-A} Interaction on \texorpdfstring{$Y$}{Y}}
\label{sec:nonconstant-homogeneity}

For identification purposes, the important feature of model \eqref{eq:strict-homogeneity} is that $U$ does not modify the effect of $A$ on $Y$. Consider the following weaker version of equation~\eqref{eq:strict-homogeneity} allowing non-constant effects:
\begin{equation} \label{eq:nonconstant-homogeneity}
Y = h(\boldsymbol{X}, \epsilon_Y)A + g(\boldsymbol{X}, U, \epsilon_Y),
\end{equation}
which implies that the individual treatment effect is $Y(a) - Y(a') = (a - a')h(\boldsymbol{X}, \epsilon_Y)$, and the CATE is given by $(a - a')\mathbb{E}(h(\boldsymbol{X}, \epsilon_Y) \mid \boldsymbol{X})$, for any two values $a, a' \in \mathcal{A}$. Due to the presence of $\epsilon_Y$ in $h$, effects may vary from unit to unit within the same $\boldsymbol{X}$-stratum, though not due to differences in $U$. As we now formalize, this model also identifies the CATE, so long as we rule out certain dependencies between $\epsilon_Y$ and $(Z, A)$ or $U$. 

\begin{proposition} \label{prop:nonconstant-ident}
    Under the IV Assumptions~\ref{ass:relevance}, \ref{ass:UC}(a), model \eqref{eq:nonconstant-homogeneity}, and assuming
    \begin{equation}\label{eq:nonconstant-ass}
    \mathbb{E}(h(\boldsymbol{X}, \epsilon_Y) \mid \boldsymbol{X}, Z, A) = \mathbb{E}(h(\boldsymbol{X}, \epsilon_Y) \mid \boldsymbol{X}),
    \end{equation}
    the CATE function $\mathbb{E}(h(\boldsymbol{X}, \epsilon_Y) \mid \boldsymbol{X})$ is identified via $\Psi_P(\boldsymbol{X})$.
\end{proposition}

In the proof of Proposition~\ref{prop:nonconstant-ident}, one may notice that the fact that $U$ does not feature in the additive effect function $h$ plays no role: the result would hold if individual treatment effects equaled $(a = a')h(\boldsymbol{X}, U, \epsilon_Y)$, so long as this function $h$ satisfied $\mathbb{E}(h(\boldsymbol{X}, U, \epsilon_Y) \mid \boldsymbol{X}, Z, A) = \mathbb{E}(h(\boldsymbol{X}, U, \epsilon_Y) \mid \boldsymbol{X})$. However, the latter assumption would be hard to justify in the context of a generative DAG as in Figure~\ref{fig:IV}. By contrast, if one were to assert model~\eqref{eq:nonconstant-homogeneity}, then assumption~\eqref{eq:nonconstant-ass} would hold automatically if one adopted a standard NPSEM model associated with a DAG from Figure~\ref{fig:IV}---either the NPSEM with independent errors~\citep{pearl2009} or the finest fully randomized causally interpretable structured tree graph~\citep{robins2010}.

\citet{wang2018} propose and study a similar condition in the case of a binary treatment, i.e., $\mathcal{A} = \{0,1\}$:
\begin{equation}\label{eq:wang-outcome}
    \mathbb{E}(Y(a = 1) - Y(a = 0) \mid \boldsymbol{X}, U) = \mathbb{E}(Y(a = 1) - Y(a = 0) \mid \boldsymbol{X}).
\end{equation}
This assumption directly encodes a lack of additive effect modification by $U$, after controlling for $\boldsymbol{X}$. Equivalently, in the language of the general structural equation~\eqref{eq:general-binary}, this asserts that $h(\boldsymbol{X}, U, \epsilon_Y)$ is mean-independent of $U$ given $\boldsymbol{X}$. This could hold, for instance, under model~\eqref{eq:nonconstant-homogeneity}, if we assumed that $\mathbb{E}(h(\boldsymbol{X}, \epsilon_Y) \mid \boldsymbol{X}, U)$ does not depend on $U$, which we can view as an alternative to assumption~\eqref{eq:nonconstant-ass}. The following identification result is proved in Theorem 1 of \citet{wang2018}. For completeness, we provide a self-contained proof in the Appendices.

\begin{proposition} \label{prop:wang-outcome-ident}
    Under the IV Assumptions~\ref{ass:relevance}, \ref{ass:UC}(a)\&(c), and condition \eqref{eq:wang-outcome} for binary $A$,
    the CATE function $\mathbb{E}(Y(a = 1) - Y(a = 0) \mid \boldsymbol{X})$ equals $\Psi_P(\boldsymbol{X})$.
\end{proposition}

Note that, as was true in Section~\ref{sec:strict-homogeneity}, once the CATE is identified, the marginal ATE is identified by averaging over the distribution of $\boldsymbol{X}$: here, $\mathbb{E}(Y(a = 1) - Y(a = 0)) = \mathbb{E}_P(\Psi_P(\boldsymbol{X}))$. Critically, both models \eqref{eq:nonconstant-homogeneity}\&\eqref{eq:nonconstant-ass} and \eqref{eq:wang-outcome} are less restrictive than the constant effect model \eqref{eq:strict-homogeneity}. However, these models still rule out the possibility that unmeasured variables act as effect modifiers, which may be hard to justify in practice~\citep{hernan2019}. Consider these models in the context of the REFLUX example. These models require that there are no unobserved covariates that are effect modifiers of the effect of surgery on quality of life. As we have previously noted, a strength of IV methods is that they allow for and acknowledge unobserved confounding of the $A \rightarrow Y$ effect. Under model~\eqref{eq:nonconstant-homogeneity} or~\eqref{eq:wang-outcome}, these unobserved variables cannot be effect modifiers on the additive scale, given measured covariates $\boldsymbol{X}$.

\subsubsection{No \texorpdfstring{$U$-$Z$}{U-Z} Interaction on \texorpdfstring{$A$}{A}}
\label{sec:no-a-interaction}

\citet{wang2018} also developed the following homogeneity assumption as an alternative to the assumption of no $A$-$Y$ effect modification: 
\begin{equation} \label{eq:wang-exposure}
    \mathbb{E}(A \mid \boldsymbol{X}, U, Z = 1) - \mathbb{E}(A \mid \boldsymbol{X}, U, Z = 0) \text{ does not depend on } U,
\end{equation}
Here we must assume that $U$ does not modify the effect of $Z$ on $A$ on the additive scale, given $\boldsymbol{X}$. This assumption also turns out to be sufficient for identification of the CATE under very similar conditions, as proved in Theorem 1 of~\citet{wang2018}.

\begin{proposition} \label{prop:wang-exposure-ident}
    Under the IV Assumptions~\ref{ass:relevance}, \ref{ass:UC}(c)\&(d), and condition \eqref{eq:wang-exposure} for binary $A$,
    the CATE function $\mathbb{E}(Y(a = 1) - Y(a = 0) \mid \boldsymbol{X})$ equals $\Psi_P(\boldsymbol{X})$.
\end{proposition}

For interpretation, let's again consider this assumption in the REFLUX example. For model~\eqref{eq:wang-exposure} to hold, it must be the case that there are no unobserved $A$-$Y$ confounders that modify the effect of assignment to surgery on actual exposure to surgery.

\subsubsection{A Weaker "No Unmeasured Common Effect Modifier" Condition} \label{sec:cui-cov}
\citet{cui2021} propose a slightly weaker form of conditions~\eqref{eq:wang-outcome} and~\eqref{eq:wang-exposure}, which allows for some degree of additive effect modification by unmeasured confounders on both the $Z$-$A$ and $A$-$Y$ relationships, so long as the functional forms of the two types of effect modification are suitably uncorrelated. Namely, letting $\widetilde{\tau}(\boldsymbol{X}, U) = \mathbb{E}(Y(a = 1) - Y(a = 0) \mid \boldsymbol{X}, U)$ and $\widetilde{\delta}(\boldsymbol{X}, U) = \mathbb{E}(A \mid \boldsymbol{X}, U, Z = 1) - \mathbb{E}(A \mid \boldsymbol{X}, U, Z = 0)$, they put forth the following condition:
\begin{equation}\label{eq:cui-cov}
    \mathrm{Cov}\left(\widetilde{\delta}(\boldsymbol{X}, U), \widetilde{\tau}(\boldsymbol{X}, U) \mid \boldsymbol{X}\right) = 0.
\end{equation}
It is straightforward to verify that either~\eqref{eq:wang-outcome}, which says that $\widetilde{\tau}(\boldsymbol{X}, U)$ depends only on $\boldsymbol{X}$, or~\eqref{eq:wang-exposure}, which says that $\widetilde{\delta}(\boldsymbol{X}, U)$ depends only on $\boldsymbol{X}$, will guarantee condition~\eqref{eq:cui-cov}. This weaker condition, however, is sufficient for identification of the CATE. We again provide a self-contained proof in the Appendices.

\begin{proposition} \label{prop:cui-cov-ident}
    Under the IV Assumptions~\ref{ass:relevance}, \ref{ass:UC}(c)\&(d), and condition \eqref{eq:cui-cov} for binary $A$,
    the CATE function $\mathbb{E}(Y(a = 1) - Y(a = 0) \mid \boldsymbol{X})$ equals $\Psi_P(\boldsymbol{X})$.
\end{proposition}

According to \citet{cui2021}, condition~\eqref{eq:cui-cov} ``essentially states that there is no common effect modifier by an unmeasured confounder, of the additive effect of treatment on the outcome, and the additive effect of the IV on treatment.'' While no elaboration is given beyond this, we infer the following interpretation: if $U$ can be decomposed into $U = (U_1, U_2)$, where $\widetilde{\delta}(\boldsymbol{X}, U)$ depends on $U$ only through $U_1$, and $\widetilde{\tau}(\boldsymbol{X}, U)$ depends on $U$ only through $U_2$, then condition~\eqref{eq:cui-cov} will hold assuming that $U_1 \independent U_2 \mid \boldsymbol{X}$. This provides subject matter experts a way to justify condition~\eqref{eq:cui-cov}, which otherwise might be quite abstract. Namely, one could try to argue that unobserved confounders modifying the $Z$-$A$ effect do not overlap with (and are conditionally independent of) those modifying the $A$-$Y$ effect.

\subsubsection{No \texorpdfstring{$Z$-$A$}{Z-A} Interaction on \texorpdfstring{$Y$}{Y}}
\label{sec:SMM}

The assumptions considered thus far have mostly explicitly (other than condition \eqref{eq:nonconstant-ass}) stipulated that unobserved variables were not effect modifiers of key effects in the IV framework. Alternatively, it is possible to identify causal effects by restricting effect modification by the instrument $Z$. For example, for a binary treatment, we might assume that $E(Y(a=1) - Y(a=0) \mid \boldsymbol{X}, Z=1, A=a) = E(Y(a=1) - Y(a=0) \mid \boldsymbol{X}, Z=0, A=a)$ for $a = 0,1$~\citep{Hernan:2006}. Here, for the treated group $(A=1)$ and the control group $(A=0)$, the causal effect is unrelated to the level of the instrument. This homogeneity assumption has typically been formulated in a somewhat weaker manner, in the context of an additive structural mean model (SMM)~\citep{robins1989, robins1994b}. Concretely, without loss of generality assuming $0 \in \mathcal{A}$, consider the following SMM:
\begin{equation} \label{eq:SMM}
    \mathbb{E}(Y - Y(a = 0) \mid \boldsymbol{X}, Z, A)  = r(\boldsymbol{X})A,
\end{equation}
for some function $r: \mathcal{X} \to \mathbb{R}$. By consistency, model \eqref{eq:SMM} is equivalent to 
\[\mathbb{E}(Y(a) - Y(a = 0) \mid \boldsymbol{X}, Z, A = a) = r(\boldsymbol{X})a, \text{ for all } a \in \mathcal{A},\]
from which it is clear that the SMM assumes no effect modification by $Z$ on the additive effect of $A$ on $Y$ relative to baseline $a = 0$, given $\boldsymbol{X}$, among those treated with $A = a$. This has also been referred to as the ``no current treatment interaction'' assumption for SMMs, which states that the causal effect of one level of the intervention (compared to baseline) is identical among subjects who actually received that level, regardless of assigned treatment \citep{Hernan:2006}. In this model, $r(\boldsymbol{X})a = \mathbb{E}(Y(a) - Y(a = 0) \mid \boldsymbol{X}, A = a)$ represents the conditional average treatment effect among the treated (CATT) for each $a \in \mathcal{A}$, e.g., if $\mathcal{A} = \{0,1\}$,
\[r(\boldsymbol{X}) = \mathbb{E}(Y(a = 1) - Y(a = 0) \mid \boldsymbol{X}, A = 1).\] Note that once the CATT is identified, the marginal ATT is identified by averaging over the distribution of $\boldsymbol{X}$ among those with $A = 1$. While the ATE measures the average difference in outcomes when all individuals in the study population are assigned to treatment versus when all individuals are assigned to control, the ATT measures the average difference in outcomes among those individuals in the population that were actually exposed to the treatment. In REFLUX, the ATT would be the average difference in quality of life among those patients that were actually exposed to surgery.

Model~\eqref{eq:SMM} together with the IV assumptions yields identification of the causal effect $r(\boldsymbol{X})$, as we now formalize.
\begin{proposition} \label{prop:SMM-ident}
    Under the IV Assumptions~\ref{ass:relevance}, \ref{ass:UC}(a), and model \eqref{eq:SMM}
    the CATT function $r(\boldsymbol{X})$ equals $\Psi_P(\boldsymbol{X})$.
\end{proposition}

As is the case for the CATE in the previous identification schemes, the CATT can be marginalized to obtain a population-level effect. Specifically, under the SMM~\eqref{eq:SMM}, we can point identify
\[\mathbb{E}(Y(a) - Y(a = 0) \mid A = a) = \mathbb{E}_P(\Psi_P(\boldsymbol{X}) \mid A = a)a,\]
the average treatment effect among the treated (ATT) for the $A = a$ group. For a binary treatment, this yields the usual ATT, $\mathbb{E}(Y(a = 1) - Y(a = 0) \mid A = 1) = \mathbb{E}_P(\Psi_P(\boldsymbol{X}) \mid A = 1)$.

To convey intuition for this homogeneity assumption~\eqref{eq:SMM}, we again use the REFLUX example from above. In words, since in this trial we may choose (by randomization) $\boldsymbol{X}$ to be empty in the core IV assumptions, the SMM for such a choice would imply that the treatment effect among those patients exposed to surgery is unrelated to the instrument (marginal over the distribution of $\boldsymbol{X} \mid A = 1$). That is, the effect of surgery among those who had undergone surgery would be the same regardless of trial arm assignment. Mathematically, with $\boldsymbol{X} = \emptyset$ and $\mathcal{A} = \{0,1\}$, model~\eqref{eq:SMM} simplifies to
\begin{equation}\label{eq:SMM-noX}
    \mathbb{E}(Y(a = 1) - Y(a = 0) \mid Z, A = 1) = \mathbb{E}(Y(a = 1) - Y(a = 0) \mid A = 1).
\end{equation} Alternatively, one could choose to incorporate measured baseline covariates $\boldsymbol{X}$---recall that this will not typically affect the validity of the IV Assumptions~\ref{ass:relevance}--\ref{ass:ER}---if model~\eqref{eq:SMM} were deemed more plausible.
As noted in \citet{hernan2019}, regardless of the incorporation of $\boldsymbol{X}$, the SMM model~\eqref{eq:SMM} is both untestable and not particularly intuitive.

Critically, applied investigators analyzing clinical trial data may mistakenly interpret Proposition~\ref{prop:SMM-ident} as guaranteeing identification of the ATT via the covariate-less IV functional $\frac{\mathbb{E}_P(Y \mid Z = 1) - \mathbb{E}_P(Y \mid Z = 0)}{\mathbb{E}_P(A \mid Z = 1) - \mathbb{E}_P(A \mid Z = 0)}$. However, this interpretation relies on the validity of the SMM without covariates, i.e., it is correct if the homogeneity assumption happens to hold \emph{unconditionally}. In most applications, though, it is likely that there are at least a few effect modifier variables. For example, the effect of surgery might vary with baseline frailty, since healthier patients might benefit more from surgery as they are less likely to experience post-operative infections. If we condition on sufficient effect modifiers $\boldsymbol{X}$ such that $Z$ itself exhibits no residual effect modification among the treated (i.e., model~\eqref{eq:SMM} is valid), analysts can recover the CATT by Proposition~\ref{prop:SMM-ident}, and thus the ATT by marginalizing.
That is, in a randomized study, conditioning on $\boldsymbol{X}$ may be useful to render the SMM more plausible, in that homogeneity would need to hold within strata of $\boldsymbol{X}$ rather than unconditionally. To be sure, mathematically, model~\eqref{eq:SMM} does not imply model~\eqref{eq:SMM-noX}: under~\eqref{eq:SMM}, we have $\mathbb{E}(Y(a = 1) - Y(a = 0) \mid Z=z, A = 1) = \mathbb{E}(r(\boldsymbol{X}) \mid Z = z, A = 1)$, which very well may depend on the level $z \in \mathcal{Z}$ since in general $Z \not \independent \boldsymbol{X} \mid A$; graphically, in Figures~\ref{fig:IVa} and~\ref{fig:IVb}, the path $Z \rightarrow A \leftarrow \boldsymbol{X}$ is open when conditioning on the collider $A$.

One may wonder whether and under what conditions a SMM such as model~\eqref{eq:SMM} may also identify the CATE. This was mostly resolved in Theorem 1 of \citet{Hernan:2006}, in which they showed that for binary instrument and treatment, if we also assume no effect modification by $Z$ on the effect of $A$ on $Y$ among those treated with $A = 0$, then the CATT and CATE coincide---by Proposition~\ref{prop:SMM-ident}, the CATE then equals $\Psi_P(\boldsymbol{X})$. We reproduce the relevant result here.

\begin{proposition}\label{prop:hernan}
    For binary $Z$ and $A$, under the IV Assumptions~\ref{ass:relevance}, \ref{ass:UC}(a), model~\eqref{eq:SMM}, and
    \begin{equation*} 
    \mathbb{E}(Y(a = 1) - Y\mid \boldsymbol{X}, Z, A)  = s(\boldsymbol{X})(1 - A),
    \end{equation*}
    for some function $s: \mathcal{X} \to \mathbb{R}$, it holds that 
    \[r(\boldsymbol{X}) = s(\boldsymbol{X}) = \mathbb{E}(Y(a = 1) - Y(a = 0) \mid \boldsymbol{X}) = \Psi_P(\boldsymbol{X}).\]
\end{proposition}

Two alternatives to the additive SMM~\eqref{eq:SMM} are also worth mentioning. First, it is possible to formulate a SMM on an alternative (e.g., multiplicative, logistic) scale \citep{robins1994b, robins2004}. In such cases, the CATT or CATE may be identified, but will no longer be equal to $\Psi_P(\boldsymbol{X})$. Second, even on the additive scale, one can replace $r(\boldsymbol{X})A$ in the right-hand side of \eqref{eq:SMM} with an arbitrary parametric model $t(\boldsymbol{X}, Z, A; \boldsymbol{\beta}_0)$, for some unknown $\boldsymbol{\beta}_0 \in \mathbb{R}^p$, such that $t(\boldsymbol{X}, Z, 0; \boldsymbol{\beta}) = 0$ for all $\boldsymbol{\beta}$. Under some estimability conditions (i.e., $\geq p$ estimating equations), one can identify and estimate $\boldsymbol{\beta}_0$ based on the equation
\[0 = \mathrm{Cov}\left(Z, Y - t(\boldsymbol{X}, Z, A; \boldsymbol{\beta}_0) \mid \boldsymbol{X}\right),\]
derived as in the proof of Proposition~\ref{prop:SMM-ident}. Given that identification diverges substantially from the other cases, and how rarely these assumptions are invoked in practice, we will not pursue estimation under either of these alternative schemes.

\subsubsection{The Role of Effect Modifiers in Homogeneity Assumptions}

One key insight that comes from assessing the various homogeneity assumptions is that, in order for these to be plausible, any effect modifiers that exist must be observed and conditioned on appropriately. Inspecting, for instance, model~\eqref{eq:wang-outcome} in Section~\ref{sec:nonconstant-homogeneity}, we see that any unobserved confounders $U$ cannot interact with the causal effect of $A$ on the additive scale, given $\boldsymbol{X}$---only the measured confounders/covariates $\boldsymbol{X}$ may modify the effect of $A$ in this manner. Therefore, including and conditioning on effect modifiers in $\boldsymbol{X}$ is essential. Again, the REFLUX application is illustrative. As we noted, for REFLUX, unconfoundedness (i.e., core IV Assumption~\ref{ass:UC}) does not require conditioning on $\boldsymbol{X}$. This is one strength of the IV framework in the context of randomized trials. However, identification under the various homogeneity assumptions requires an absence of unobserved effect modifiers. In short, investigators should condition on a broad set of covariates in $\boldsymbol{X}$ in case those variables are effect modifiers. Conditioning on a diverse collection of covariates $\boldsymbol{X}$ increases even more so the plausibility of the constant effect model~\eqref{eq:strict-homogeneity}, since we would only require the treatment effect to be constant within each $\boldsymbol{X}$-stratum. In sum, in applied settings like REFLUX, while conditioning on $\boldsymbol{X}$ is unnecessary for Assumption~\ref{ass:UC}, it does make homogeneity assumptions more plausible. Moreover, outside of the context of randomized trials, in choosing and measuring $\boldsymbol{X}$, one should be mindful of including both $Z$-$(A, Y)$ confounders, but also key effect modifiers.

\subsection{Monotonicity} 
\label{sec:monotonicity}

We now turn to the next major class of assumptions used for point identification, based on ``monotonicity''. One feature of these assumptions is that we do not require any assumptions about effect homogeneity. On the other hand, this will come at some cost in terms of the interpretability and relevance of the target causal estimand we can identify. As we will highlight, the usefulness of the estimand is a matter of some debate in the literature.

\subsubsection{Deterministic Monotonicity}

Perhaps the most common structural assumption used in the applied literature for binary instrument and treatment (i.e., $\mathcal{Z} = \mathcal{A} = \{0,1\}$) is one referred to as monotonicity, which has the following form: $A(z = 1) \geq A(z = 0)$, i.e., if $A(z = 0) = 1$ then $A(z = 1) = 1$. The celebrated result of Imbens \& Angrist \citep{imbens1994}, in their Theorem 1, is that under monotonicity, the following causal estimand is identified:
\begin{equation}
\label{eq:LATE}
  \mathbb{E}\left(Y(a = 1) - Y(a = 0) \mid A(z = 1) > A(z = 0), \boldsymbol{X}\right),
\end{equation} 
the so-called conditional local average treatment effect (CLATE).
\begin{proposition} \label{prop:LATE}
    For binary $Z$ and $A$, under the IV Assumptions~\ref{ass:relevance}, \ref{ass:UC}(b), \ref{ass:ER}, monotonicity, i.e., $A(z =1) \geq A(z = 0)$,
    the CLATE function equals $\Psi_P(\boldsymbol{X})$.
\end{proposition}

A marginal version of the CLATE, the local average treatment effect (LATE), is given by
\[\mathbb{E}(Y(a = 1) - Y(a = 0) \mid A(z = 1) > A(z = 0)) = \frac{\mathbb{E}_P(\mu_1(\boldsymbol{X}) - \mu_0(\boldsymbol{X}))}{\mathbb{E}_P(\lambda_1(\boldsymbol{X}) - \lambda_0(\boldsymbol{X}))},\]
with identification following the same argument as used in the proof of Proposition~\ref{prop:LATE}, under the same assumptions. Typically, the CLATE and LATE are interpreted by stratifying the study observations into four groups or principal strata~\citep{Angrist:1996,Frangakis:2002}.  That is, when both the instrument and treatment are binary, we can classify subjects as \textit{compliers} (i.e., $A(z = 1) > A(z = 0)$), \textit{always-takers} (i.e., $A(z = 1) = A(z = 0) = 1$), \textit{never-takers} (i.e., $A(z = 1) = A(z = 0)=0$) or \textit{defiers} (i.e., $A(z = 1) < A(z = 0)$). With these labels, the estimand in  \eqref{eq:LATE} represents the causal effect of exposure $A$ on $Y$ among compliers, i.e., those for whom $A(z = 1) = 1$ and $A(z = 0) = 0$---for this reason, the LATE is also often referred to as the \textit{complier} average causal effect. By Proposition~\ref{prop:LATE}, such complier-specific effects are identified when $Z$ is an instrument and the monotonicity assumption holds. In REFLUX, the LATE represents the average affect of surgery among those patients who receive surgery because they were assigned the surgery arm of the trial. Alternatively, in the langauge of principal strata, the LATE is the effect of surgery among those patients that are compliers.

The LATE is often targeted in trials and observational studies with measured IVs when the monotonicity assumption is thought to hold. However, since the monotonicity assumption rules out the presence of defiers, i.e., those for whom  $A(z = 1) = 0$ and $A(z = 0) = 1$, one must be able to argue that no defiers exist on subject matter knowledge. In many contexts, the monontoncity assumption is plausible. In REFLUX, defiers are patients who are only exposed to surgery if not assigned to the surgery arm or are given medical manage if assigned to the surgery arm. It would seem  unlikely that there are patients who would only receive surgery if randomly assigned to medical management and vice versa. Moreover, in REFLUX, the noncompliance decision is strongly dictated by the attending physician, so defiers would generally only be present when physicians act directly contrary to the trial protocol.

In some cases one can rule out defiers by design. Some randomized trial designs include a ``no-contamination'' restriction, or satisfy ``one-sided non-compliance'', which is defined as the absence of off-protocol use of the intervention among controls, such that $Z = 0$ implies $A = 0$; in counterfactual notation, $A(z = 0) = 0$ \citep{Cuzick:2007}. This restriction thus precludes the possibility of always-takers in addition to defiers. In REFLUX, this would imply that patients in the medical management arm were not allowed to cross over to the surgery arm of the trial. 
In addition, the no-contamination restriction implies the no current treatment interaction assumption for SMMs, since it reduces the number of causal effects to one among those who received the intervention. Mathematically, since $A = 1$ implies $Z = 1$, equation~\eqref{eq:SMM} holds vacuously, as $\mathbb{E}(Y(a = 1) - Y(a = 0) \mid \boldsymbol{X}, Z, A = 1) = \mathbb{E}(Y(a = 1) - Y(a = 0) \mid \boldsymbol{X}, Z = 1, A = 1)$ is a function only of $\boldsymbol{X}$. This indicates that under one-sided non-compliance, by Propositions~\ref{prop:SMM-ident} and~\ref{prop:LATE}, the CLATE coincides with the CATT \citep{frolich2013}.

\begin{proposition} \label{prop:one-sided}
    For binary $Z$ and $A$, under the IV Assumptions~\ref{ass:relevance}, \ref{ass:UC}(b), \ref{ass:ER}, and one-sided non-compliance, i.e., $A(z = 0) = 0$,
    the CATT function equals $\Psi_P(\boldsymbol{X})$.
\end{proposition}

Under one-sided non-compliance, compliers correspond to those with $A(z = 1) = 1$. With this in mind, the equality of the CLATE and CATT implied by Proposition~\ref{prop:one-sided} is not too surprising. Intuitively, the treated subpopulation is representative of compliers (in terms of the counterfactual outcome distribution) within levels of $\boldsymbol{X}$, by unconfoundedness.

Under monotonicity, one must in general be content with the fact that complier-specific effects do not provide analysts with any information about the causal effects among the always-taker and never-taker subpopulations. Moreover, the subpopulation of compliers is generally not identified, so that it is often unclear to which subjects the LATE applies. Nonetheless, the \textit{distribution} of covariates $\boldsymbol{X}$ among compliers is identified, and moreover, in some settings it may be possible to accurately predict who compliers are, and obtain tight bounds on causal effects in identifiable subgroups~\citep{kennedy2020b}.

An additional applied example is useful to understand how monotonicity alters the IV estimand. Imagine a randomized trial where study participants are randomly encouraged to exercise for four weeks, i.e., $Z = 1$ if encouraged, and $Z = 0$ if not encouraged, and the outcome $Y$ of the study is lung volume after four weeks. In the study, the exposure $A$ is dichotomized into any or no exercise sessions. Monotonicity, in this case, would assert that there are no subjects that would exercise when not encouraged to do so, and not exercise if encouraged. Under monotonicity, the LATE represents the effect of exercise among the subset of participants that exercised \textit{because} they were encouraged to exercise. The estimated LATE tells the analyst little about the effect of exercise among the always-takers and never-takers---in this study the always-takers are the participants that exercise regardless of encouragement, and the never-takers are the participants that do not exercise even if encouraged. Specifically, we might expect that the effect of exercise is smaller among the always-takers and larger among the never-takers, but this cannot be assessed from the observed data when one only assumes monotonicity.  

This last point may be unsettling from a policy perspective, and the LATE estimand has engendered considerable controversy in the literature, where it is often criticized as an overly narrow and irrelevant target quantity~\citep{Deaton:2010,swanson2014think} (though see~\citet{Imbens:2010,imbens2014instrumental} for defenses of the LATE estimand). In other cases, the monotonicity assumption itself has been criticized as implausible in some applications. In particular, monotonicity violations may arise when the instrument is not delivered in a uniform way to all subjects~\citep{swanson2014think,swanson2017challenging}. This critique led to what is known as the stochastic monotonicity assumption~\citep{small2017instrumental}, which we review next.



\subsubsection{Stochastic Monotonicity}\label{sec:stochastic-mon}

With binary instrument and treatment, the deterministic monotonicity assumption asserts that $A(z = 1) \geq A(z = 0)$, meaning that the probability of any subject being a defier is exactly zero. In many settings, this can be overly strong, and one can imagine situations where there is a positive probability of defiance given that the instrument delivers encouragement to treatment exposure differentially. For example, excess travel time to a specialized health care facility compared to a normal health care facility has been used as an IV for whether a person receives a certain type of care in many health studies e.g., \citet{mcclellan1994does,Lorch:2012}.  That is, patients that are close to a specialized hospital may be more likely to receive care at that hospital. Instruments of this type, however, affect the decision for actual treatment received differentially across patients, and this may result in patients that are defiers with respect to the IV. In the REFLUX example, such violations are unlikely, since the assignment to each treatment arm effectively provides a uniform encouragement to each arm of the study.

In order to deal with potential violations of strict deterministic monotonicity, some authors have proposed a weaker version of this assumption, which has been termed \textit{stochastic monotonicity}~\citep{dinardo2011, ramsahai2012, small2017instrumental}:
\begin{equation}\label{eq:stochastic-mon}
    P[A = 1 \mid \boldsymbol{X}, U, Z = 1] \geq P[A = 1 \mid \boldsymbol{X}, U, Z = 0], \text{ almost surely},
\end{equation}
asserting that within strata of the observed confounders $\boldsymbol{X}$ and unobserved confounders $U$, the probability of treatment with $A = 1$ is at least as high when $Z = 1$ compared to when $Z = 0$.
\citet{small2017instrumental} show that under the stochastic monotonicity assumption~\eqref{eq:stochastic-mon}, the function $\Psi_P(\boldsymbol{X})$ continues to have an interpretation as a conditional treatment-outcome causal effect. Concretely, defining
\[w(\boldsymbol{X}, U) = P[A = 1 \mid \boldsymbol{X}, U, Z = 1] - P[A = 1 \mid \boldsymbol{X}, U, Z = 0],\]
which is non-negative if and only if~\eqref{eq:stochastic-mon} holds (note also that $w \equiv \widetilde{\delta}$ defined in Section~\ref{sec:cui-cov}), and under conditions similar to IV Assumptions~\ref{ass:relevance}--\ref{ass:ER},
\begin{equation}  \label{eq:stochastic-ident} 
\Psi_P(\boldsymbol{X}) = \frac{\mathbb{E}(\{Y(a = 1) - Y(a = 0)\}w(\boldsymbol{X}, U) \mid \boldsymbol{X})}{\mathbb{E}(w(\boldsymbol{X}, U) \mid \boldsymbol{X})}.
\end{equation}
\citet{small2017instrumental} call this quantity the strength-of-IV weighted conditional average treatment effect (SIV-WCATE), as it can be interpreted as marginalizing the fully conditional CATE, $\mathbb{E}(Y(a = 1) - Y(a = 0) \mid \boldsymbol{X}, U)$, over a distribution $Q$ for $U \mid \boldsymbol{X}$ proportional to $w(\boldsymbol{X}, U) p(U \mid \boldsymbol{X})$. In fact, these authors show that this interpretation is valid regardless of whether the instrument $Z$ itself has a well-defined causal effect, so long as the stochastic monotonicity assumption~\eqref{eq:stochastic-mon} holds. This means that the counterfactuals $A(z)$ and $Y(z, a)$ need not even be well-defined, which may occur if: (i) $Z$ is an ``intensity-preserving'' proxy for a manipulable causal IV $Z^*$, or (ii) $Z$ (or $Z^*$ for which $Z$ is a proxy) is not manipulable---see \citet{small2017instrumental} for details.
We summarize their findings in the following result.

\begin{proposition}\label{prop:stochastic-ident}
    Suppose $Z, A \in \{0,1\}$, and counterfactuals $Y(a)$ for $a \in \{0,1\}$ are well-defined such that consistency holds, i.e., $Y(A) = Y$. Assume that unmeasured confounders $U$ satisfy
    $A \independent Y(a) \mid \boldsymbol{X}, U, Z$,
    and also that the effect of $Z$ is unconfounded, in that $Z \independent Y(a) \mid \boldsymbol{X}, U$ and $Z \independent U \mid \boldsymbol{X}$. Then under Assumption~\ref{ass:relevance}, equation~\eqref{eq:stochastic-ident} holds. In particular, under stochastic monotonicity~\eqref{eq:stochastic-mon}, $\Psi_P(\boldsymbol{X})$ has the SIV-WCATE interpretation $\mathbb{E}_{Q}(\mathbb{E}(Y(a = 1) - Y(a = 0) \mid \boldsymbol{X}, U)\mid \boldsymbol{X})$.
\end{proposition}

Proposition~\ref{prop:stochastic-ident} is a substantial generalization of Proposition~\ref{prop:LATE}, and assigns a causal interpretation to $\Psi_P(\boldsymbol{X})$ in a variety of scenarios beyond having a well-defined manipulable IV satisfying strict deterministic monotonicity---if indeed $Z$ is manipulable, the conditions of Proposition~\ref{prop:stochastic-ident} would be satisfied under Assumptions~\ref{ass:relevance}, \ref{ass:UC}(c)\&(d), and \ref{ass:ER}. As in the other identification paradigms, a marginal version of the SIV-WCATE, the strength-of-IV weighted average treatment effect (SIV-WATE) is also identified by the appropriate integration. Here,
\[\frac{\mathbb{E}(\{Y(a = 1) - Y(a = 0)\}w(\boldsymbol{X}, U))}{\mathbb{E}(w(\boldsymbol{X}, U))}= \frac{\mathbb{E}_P\left(\mu_1(\boldsymbol{X}) - \mu_0(\boldsymbol{X})\right)}{\mathbb{E}_P\left(\lambda_1(\boldsymbol{X}) - \lambda_0(\boldsymbol{X})\right)}.\]
The left-hand side of this display can be written as $\mathbb{E}_{\mathcal{Q}}(\mathbb{E}(Y(a = 1) - Y(a = 0) \mid \boldsymbol{X}, U))$, i.e., an integral of the fully conditional CATE (given $\boldsymbol{X}$ and $U$) over a distribution $\mathcal{Q}$ for $(\boldsymbol{X}, U)$ proportional to $w(\boldsymbol{X}, U)p(\boldsymbol{X}, U)$.

Interpretation of the SIV-WCATE (and SIV-WATE) is considerably more complex than that of the CLATE (and LATE). The SIV-WCATE is a weighted average of treatment effects, where the weight is based on groups or strata defined by $\boldsymbol{X}$ and $U$. Specifically, the conditional treatment effect within a $(\boldsymbol{X}, U)$-stratum is weighted according to the actual size of that group times the effect of the instrument on the exposure in that stratum, i.e., in that group, the gap between the probability of taking treatment when encouraged by the IV to take treatment (i.e., $Z = 1$) compared to when not encouraged to take treatment (i.e., $Z = 0$). In other words, $Q$ weights each subject by how strongly the IV is associated with treatment in the subgroups defined by $(\boldsymbol{X}, U)$, hence the term \textit{strength-of-IV weighting} in the naming of these estimands. As a means of understanding the hypothetical population to which this effect ``generalizes'', \citet{small2017instrumental} recommend characterizing how the distribution of the observed covariates for $Q$ compares to the distribution of covariates for the unweighted population.

An applied example is useful at this point to better explain the interpretation of the SIV-WCATE. In comparative effectiveness research, analysts often use a physician's preference for a specific course of treatment as an IV. Here, a physician's preference is a factor that causes some patients to be exposed to a specific treatment, but arguably has no direct effect on the outcome \citep{brookhart2007preference,brookhart2006evaluating}.  One study, for example, uses a surgeon's preference for operative management for emergency conditions as an IV for receipt of surgery \citep{keeleegsiv2018}.  They call this IV \textit{tendency to operate} (TTO). In this example, $Z=1$ for patients that receive care from a high TTO surgeon, and $Z=0$ for patients that receive care from a low TTO surgeon. Likewise, $A=1$ indicates receiving emergency surgery, and $A=0$ indicates not receiving surgery. Under deterministic monotonicity, the LATE represents the effect of surgery among those patients that receive surgery because they receive care from a high TTO surgeon. \citet{swanson2015definition} argue, however, that deterministic monotoncity is unlikely to hold in this type of application. 

\citet{small2017instrumental} outline that stochastic monotonicity is formulated precisely for this type of application. To interpret this assumption and the SIV-WCATE in the TTO example, the analyst must first conceptualize a set of unmeasured confounders $U$. For this type of application, \citet{small2017instrumental} suggest defining $U$ to be the vector of all patient characteristics that systematically affect treatment or potential outcomes, beyond those already measured and included in $\boldsymbol{X}$. This way, patients who share the same $\boldsymbol{X}$ and $U$ can be considered patients of the same ``type''. Stochastic monotonicity asserts that, for patients of the same type, the probability of surgery is higher for those patients with surgeons with a high TTO compared to those patients with surgeons with a low TTO. Finally, the SIV-WCATE is a weighted average of treatment effects that puts more weight on patient types whose treatment choices are more influenced by TTO.  To get a better understanding of the applicability of the SIV-WCATE, researchers would want to compare the distribution of the weighted population to the overall study population.

This more detailed discussion is, we think, revealing. While stochastic monotonicity is more plausible than deterministic monotonicity for a wide range of applications, interpretation of the SIV-WCATE is considerably more difficult. See Section 6 of \citet{small2017instrumental} for additional discussion of the interpretation of this estimand.

\subsection{Other Identification Schemes}

We have aimed to provide a comprehensive review of the most widely used point identification schemes in the IV literature. In the results of Sections~\ref{sec:homogeneity} and~\ref{sec:monotonicity}, we have seen that the function $\Psi_P(\boldsymbol{X})$ identifies a particular kind of treatment-outcome causal effect, depending on the structural assumptions invoked. This sets the stage for a relatively all-encompassing statistical estimation framework, which we will present in Section~\ref{sec:estimation}. Nonetheless, there are alternative identification schemes that have been proposed, for which the resulting identified conditional causal contrast does not equal $\Psi_P(\boldsymbol{X})$. We briefly highlight two such proposals here (with no guarantee of comprehensiveness), though we will not pursue estimation under these paradigms in Section~\ref{sec:estimation}.

First, \citet{tchetgen2013} put forth a ``homogeneous selection bias'' assumption as an alternative to the additive SMM~\eqref{eq:SMM}. Namely, they show in their Lemma 1 that the CATT is identified under Assumptions~\ref{ass:relevance}, \ref{ass:UC}(a), ~\ref{ass:ER}, and assuming
\[\mathbb{E}(Y(a = 0) \mid \boldsymbol{X}, A = 1, Z) - \mathbb{E}(Y(a = 0) \mid \boldsymbol{X}, A = 0, Z) = u(\boldsymbol{X}),\]
for some function $u$, which is to say that the amount of unmeasured $A$-$Y$ confounding on the additive scale (given $\boldsymbol{X}$ and $Z$) is balanced with respect to $Z$. Identification of the CATT is not as straightforward under homogeneous selection bias, and does not correspond to $\Psi_P(\boldsymbol{X})$. For further details, we refer the reader to \citet{tchetgen2013}.

Second, \citet{liu2020} were concerned with estimation of the marginal ATT parameter $\mathbb{E}(Y(a = 1) - Y(a = 0) \mid A = 1)$. Under Assumptions~\ref{ass:relevance}, \ref{ass:UC}(a), and \ref{ass:ER}, they developed necessary and sufficient conditions for identifying the distribution of $Y(a = 0)$ given $(\boldsymbol{X}, Z, A)$, which in turn is enough to identify the ATT by consistency. Again, the identification formula no longer directly involves $\Psi_P(\boldsymbol{X})$, and invokes the ``extended propensity score'' $P[A = 1 \mid \boldsymbol{X}, Z, Y(a = 0)]$. As a consequence, the doubly robust-style estimator proposed in \citet{liu2020} relies on a parametric model for this nuisance function.

\section{Efficient Estimation and Inference} 
\label{sec:estimation}

In this section, we present a unifying estimation framework for targeting causal contrasts identified by the schemes laid out in Sections~\ref{sec:homogeneity} and~\ref{sec:monotonicity}. We will focus our analysis on the case of a binary instrument $Z \in \mathcal{Z} = \{0,1\}$ (see \citet{kennedy2019} for related methods for continuous instruments). Our goal will be to exploit nonparametric efficiency theory to construct efficient estimators of target functionals---based on their influence functions---that are robust and amenable to flexible estimation of component nuisance functions; see \citet{kennedy2022} for a broader review of this approach. Throughout, we will assume we observe an iid sample $(O_1, \ldots, O_n)$, where an arbitrary observation is $O = (\boldsymbol{X}, Z, A, Y) \sim P$.

To summarize the relevant identification results, the function $\Psi_P(\boldsymbol{X}) = \frac{\mu_1(\boldsymbol{X}) - \mu_0(\boldsymbol{X})}{\lambda_1(\boldsymbol{X}) - \lambda_0(\boldsymbol{X})}$ (recalling Lemma~\ref{lemma:binary}) represents the CATE under the assumptions of Propositions~\ref{prop:homogeneity-ident}, \ref{prop:nonconstant-ident}, \ref{prop:wang-outcome-ident}, \ref{prop:wang-exposure-ident}, \ref{prop:cui-cov-ident}, or \ref{prop:hernan}; the CATT the assumptions of Proposition~\ref{prop:SMM-ident} or~\ref{prop:one-sided}; the CLATE under the assumptions of Proposition~\ref{prop:LATE}; and the SIV-WCATE under the assumptions of Proposition~\ref{prop:stochastic-ident}. Our target parameters will be marginal and conditional effects that result in each of these identification paradigms. Concretely, let $V = \mathcal{C}(\boldsymbol{X})$ represent some coarsening of $\boldsymbol{X}$, e.g., $V$ a subset of the covariates, and let $\overline{V}$ be such that $\boldsymbol{X}$ is in one-to-one correspondence with $(V, \overline{V})$. Moreover, let $\gamma(\boldsymbol{X}) = \mu_{1}(\boldsymbol{X}) - \mu_0(\boldsymbol{X})$ and $\delta(\boldsymbol{X}) = \lambda_{1}(\boldsymbol{X}) - \lambda_0(\boldsymbol{X})$, such that $\Psi_P(\boldsymbol{X}) = \frac{\gamma(\boldsymbol{X})}{\delta(\boldsymbol{X})}$. Then under the assumptions of Propositions~\ref{prop:homogeneity-ident}, \ref{prop:nonconstant-ident}, \ref{prop:wang-outcome-ident}, \ref{prop:wang-exposure-ident}, \ref{prop:cui-cov-ident}, or \ref{prop:hernan}, the average treatment effect within levels of $V$, $\mathbb{E}(Y(a = 1) - Y(a = 0) \mid V)$, is given by
\begin{equation}\label{eq:CATE}
\xi_P(V) \coloneqq \mathbb{E}_P(\Psi_P(\boldsymbol{X}) \mid V);
\end{equation}
under the assumptions of Proposition~\ref{prop:SMM-ident} or~\ref{prop:one-sided}, the average treatment effect among the treated within levels of $V$, $\mathbb{E}(Y(a = 1) - Y(a = 0) \mid V, A = 1)$, is given by
\begin{equation} \label{eq:CATT}
\zeta_P(V) \coloneqq \mathbb{E}_P(\Psi_P(\boldsymbol{X}) \mid V, A = 1);
\end{equation}
under the assumptions of Proposition~\ref{prop:LATE}, the local average treatment effect within levels of $V$, $\mathbb{E}(Y(a = 1) - Y(a = 0) \mid V, A(z = 1) > A(z = 0))$, is given by
\begin{equation}\label{eq:CLATE}
    \chi_P(V) \coloneqq \frac{\mathbb{E}_P(\gamma(\boldsymbol{X})\mid V)}{\mathbb{E}_P(\delta(\boldsymbol{X})\mid V)};
\end{equation}
and under the assumptions of Proposition~\ref{prop:stochastic-ident}, the strength-of-IV weighted average treatment effect within levels of $V$, i.e., the mean of $\mathbb{E}(Y(a = 1) - Y(a = 0) \mid \boldsymbol{X}, U)$ over a distribution $Q$ for $U, \overline{V} \mid V$ proportional to $w(\boldsymbol{X}, U) p(U, \overline{V} \mid V)$, is also given by \eqref{eq:CLATE}.

Taking $V = \emptyset$ yields the marginal ATE, ATT, LATE, and SIV-WATE in the above paradigms, whereas more refined choices of $V$ enable assessment of heterogeneous treatment effects. In Section~\ref{sec:est-marginal} below, we focus on estimation of marginal quantities by taking $V = \emptyset$---estimation of subgroup effects in the case that $V$ is discrete is entirely analogous.
In Section~\ref{sec:est-project}, we propose a projection approach for estimating heterogeneous effects in the case that $V$ is high-dimensional and/or contains continuous components. Proofs for all results can be found in the Appendices.

\subsection{Marginal Parameters} \label{sec:est-marginal}
Define $\xi(P) = \mathbb{E}_P(\Psi_P(\boldsymbol{X}))$, $\zeta(P) = \mathbb{E}_P(\Psi_P(\boldsymbol{X}) \mid A = 1)$, and $\chi(P) = \frac{\mathbb{E}_P(\gamma(\boldsymbol{X}))}{\mathbb{E}_P(\delta(\boldsymbol{X}))}$, specializing definitions~\eqref{eq:CATE}, \eqref{eq:CATT}, and~\eqref{eq:CLATE}, respectively, to the case where $V = \emptyset$. In the following result, we present von Mises expansions for these three marginal quantities of interest, which at once will (a) yield their nonparametric influence functions and corresponding nonparametric efficiency bounds, and (b) suggest robust estimators that can achieve their efficiency bounds under relatively weak conditions on component nuisance function estimators \citep{bkrw1993, vandervaart2000, kennedy2022}. Going forward, define the composite nuisance function $\rho(\boldsymbol{X}) = P[A = 1 \mid \boldsymbol{X}] = \pi_0(\boldsymbol{X})\lambda_0(\boldsymbol{X}) + \pi_1(\boldsymbol{X})\lambda_1(\boldsymbol{X})$, and the functionals $\Gamma(P) = \mathbb{E}_P(\gamma(\boldsymbol{X}))$ and $\Delta(P) = \mathbb{E}_P(\delta(\boldsymbol{X}))$.

\begin{lemma}\label{lemma:VM-marginal}
    The marginal ATE, ATT, and LATE/SIV-WATE parameters satisfy the von Mises expansions
\[\xi(\overline{P}) - \xi(P) = - \int \dot{\xi}(o; \overline{P})\, dP(o)+ R_{\xi}(\overline{P}, P),\]
\[\zeta(\overline{P}) - \zeta(P) = - \int \dot{\zeta}(o; \overline{P})\, dP(o)+ R_{\zeta}(\overline{P}, P),\]
\[\chi(\overline{P}) - \chi(P) = - \int \dot{\chi}(o; \overline{P})\, dP(o)+ R_{\chi}(\overline{P}, P),\]
respectively, where
\[\dot{\xi}(O; P) =\frac{2Z - 1}{\delta(\boldsymbol{X})
\pi_Z(\boldsymbol{X})}\left\{Y - \mu_Z(\boldsymbol{X}) -
\Psi_P(\boldsymbol{X})(A - \lambda_Z(\boldsymbol{X}))\right\} +
\Psi_P(\boldsymbol{X}) - \xi(P),\]

\begin{align*}
    \dot{\zeta}(O; P) 
    &= \frac{\rho(\boldsymbol{X})}{P[A = 1]}\cdot
\frac{2Z - 1}{\delta(\boldsymbol{X}) \pi_Z(\boldsymbol{X})}\left\{Y -
\mu_Z(\boldsymbol{X}) - \Psi_P(\boldsymbol{X})(A -
\lambda_Z(\boldsymbol{X}))\right\} \\
& \quad \quad \quad + \frac{A}{P[A =
1]}\left(\Psi_P(\boldsymbol{X}) - \zeta(P)\right),
\end{align*}
\begin{align*}
    \dot{\chi}(O; P) 
    &= \frac{2Z - 1}{\Delta(P)
\pi_Z(\boldsymbol{X})}\left\{Y - \mu_Z(\boldsymbol{X}) - \chi(P)(A -
\lambda_Z(\boldsymbol{X}))\right\} \\
& \quad \quad \quad \quad \quad +
\frac{1}{\Delta(P)}\left(\gamma(\boldsymbol{X}) -
\chi(P)\delta(\boldsymbol{X})\right),
\end{align*}

and the remainder terms, omitting inputs for brevity, are given by
\begin{align*}
&R_{\xi}(\overline{P}; P) \\
&= \mathbb{E}_P\bigg(\frac{1}{\overline{\delta}}\bigg\{
(\overline{\pi}_1 - \pi_1)\bigg\{\frac{\overline{\mu}_1 - \mu_1 -
\frac{\overline{\gamma}}{\overline{\delta}}(\overline{\lambda}_1 - \lambda_1)}
{\overline{\pi}_1} +
\frac{\overline{\mu}_0 - \mu_0 -
\frac{\overline{\gamma}}{\overline{\delta}}(\overline{\lambda}_0 - \lambda_0)}
{1 - \overline{\pi}_1}\bigg\} \\
& \quad \quad \quad \quad \quad + (\overline{\delta} - \delta)
\left(\frac{\overline{\gamma}}{\overline{\delta}} -
\frac{\gamma}{\delta}\right)
\bigg\}\bigg), \\
&R_{\zeta}(\overline{P}; P) \\
&= \frac{1}{\overline{P}[A=1]}\mathbb{E}_P\bigg(\frac{\overline{\rho}}{\overline{\delta}}\bigg\{
(\overline{\pi}_1 - \pi_1)\bigg\{\frac{\overline{\mu}_1 - \mu_1 -
\frac{\overline{\gamma}}{\overline{\delta}}(\overline{\lambda}_1 - \lambda_1)}
{\overline{\pi}_1} +
\frac{\overline{\mu}_0 - \mu_0 -
\frac{\overline{\gamma}}{\overline{\delta}}(\overline{\lambda}_0 - \lambda_0)}
{1 - \overline{\pi}_1}\bigg\}\bigg\} \\
& \quad \quad \quad \quad \quad + \bigg(\frac{\overline{\gamma}}{\overline{\delta}}
- \frac{\gamma}{\delta}\bigg)\bigg[(\rho - \overline{\rho})+
\frac{\overline{\rho}}{\overline{\delta}}(\overline{\delta} - \delta)\bigg] \bigg) \\
& \quad \quad + \left(\zeta(\overline{P}) - \zeta(P)\right)\left(1 - \frac{P[A = 1]}{\overline{P}[A = 1]}\right), \\
& R_{\chi}(\overline{P}; P) \\
&= \left(\chi(\overline{P}) - \chi(P)\right)\left(1 - \frac{\Delta(P)}{\Delta(\overline{P})}\right)
- \frac{\chi(\overline{P})}{\Delta(\overline{P})}
\mathbb{E}_P\left((\overline{\pi}_1 - \pi_1) \left\{\frac{\overline{\lambda}_1 - \lambda_1}{\overline{\pi}_1} + \frac{\overline{\lambda}_0 - \lambda_0}{1 - \overline{\pi}_1}\right\}\right) \\
& \quad \quad \quad \quad \quad + \frac{1}{\Delta(\overline{P})} \mathbb{E}_P\left((\overline{\pi}_1 - \pi_1) \left\{\frac{\overline{\mu}_1 - \mu_1}{\overline{\pi}_1} + \frac{\overline{\mu}_0 - \mu_0}{1 - \overline{\pi}_1}\right\}\right).
\end{align*}
\end{lemma}

As mentioned, Lemma~\ref{lemma:VM-marginal} yields the nonparametric influence functions of the parameters $\xi(P)$, $\zeta(P)$, $\chi(P)$. The key observation is that the remainder terms $R_{\xi}$, $R_{\zeta}$, and $R_{\chi}$ are sums of products of errors between $\overline{P}$ and $P$, i.e., second order.
We summarize in the following result.

\begin{corollary}\label{cor:IF-marginal}
The nonparametric efficient influence functions of $\xi$, $\zeta$, and $\chi$, at $P$ are $\dot{\xi}(O; P)$, $\dot{\zeta}(O;P)$, and $\dot{\chi}(O; P)$, respectively.
\end{corollary}

As alluded to above, Lemma~\ref{lemma:VM-marginal} and Corollary~\ref{cor:IF-marginal} will provide the basis for flexible and efficient estimators of the quantities of interest. We remark first that the influence functions for $\xi(P)$ and $\chi(P)$ have appeared elsewhere. Namely, nonparametric efficiency theory and influence function-based estimators of the ATE parameter $\xi(P)$ were developed in \citet{wang2018}. These authors, however, only considered parametric models for component nuisance function estimators, whereas we will allow flexible machine learning-based estimators. Moreover, influence function-based estimation of the LATE parameter $\chi(P)$ has been well studied---see~\citet{ogburn2015, belloni2017, takatsu2023drmliv}. Second, we note that the subtle difference in the mean-of-ratio type parameter $\xi(P) = \mathbb{E}_P\left(\frac{\gamma(\boldsymbol{X})}{\delta(\boldsymbol{X})}\right)$ for the ATE, and ratio-of-mean type parameter $\chi(P) = \frac{\mathbb{E}_P\left(\gamma(\boldsymbol{X})\right)}{\mathbb{E}_P\left(\delta(\boldsymbol{X})\right)}$ for the LATE, results in substantial differences in their influence functions. This illustrates how, in general, the estimation strategy must be closely tied to the identification strategy that one pursues in a given setting.

We now propose robust estimators of the marginal causal contrasts, making use of the above results. Throughout, we will assume for simplicity that we have nuisance estimates $\widehat{P} = (\widehat{\pi}_1, \widehat{\mu}_0, \widehat{\mu}_1, \widehat{\lambda}_0, \widehat{\lambda}_1)$ trained on a separate independent sample. In practice, given one random sample, one can randomly split the data in folds to achieve the same separation, and perform cross-fitting (i.e., swap the role of the folds, repeat the procedure and average the resulting estimates) to regain efficiency \citep{bickel1988, robins2008, zheng2010, chernozhukov2018}. To simplify notation, we will go forth with just one split of the data, noting that analysis of the cross-fit procedure follows easily.

The von Mises expansions laid out in Lemma~\ref{lemma:VM-marginal} each motivate a one-step procedure to correct the bias of a plug-in estimator. For example, the first-order bias of the plug-in ATE estimator $\xi(\widehat{P})$ can be mitigated by subtracting off an estimate of this bias: \[\widehat{\xi} = \xi(\widehat{P}) + \mathbb{P}_n\left(\dot{\xi}(O; \widehat{P})\right),\]
where $\mathbb{P}_n(f) = \frac{1}{n}\sum_{i=1}^n f(O_i)$ is the empirical mean of the function $f(O)$. Indeed, we propose such a one-step estimator for both the ATE estimand $\xi(P)$ and the ATT estimand $\zeta(P)$:
\[\widehat{\xi} =
\mathbb{P}_n\left(\frac{2Z - 1}{\widehat{\delta}(\boldsymbol{X})
\widehat{\pi}_Z(\boldsymbol{X})}\left\{Y -
\widehat{\mu}_Z(\boldsymbol{X}) -
\frac{\widehat{\gamma}(\boldsymbol{X})}{\widehat{\delta}(\boldsymbol{X})}(A -
\widehat{\lambda}_Z(\boldsymbol{X}))\right\} +
\frac{\widehat{\gamma}(\boldsymbol{X})}{\widehat{\delta}(\boldsymbol{X})}\right),\]
\begin{align*}
\widehat{\zeta}
&=
\mathbb{P}_n\bigg(\frac{\widehat{\rho}(\boldsymbol{X})}
{\widehat{P}[A = 1]}\frac{2Z - 1}{\widehat{\delta}(\boldsymbol{X})
\widehat{\pi}_Z(\boldsymbol{X})}\bigg\{Y -
\widehat{\mu}_Z(\boldsymbol{X}) -
\frac{\widehat{\gamma}(\boldsymbol{X})}{\widehat{\delta}(\boldsymbol{X})}(A -
\widehat{\lambda}_Z(\boldsymbol{X}))\bigg\} \\
& \quad \quad \quad \quad +
\frac{A}{\widehat{P}[A = 1]}\frac{\widehat{\gamma}(\boldsymbol{X})}
{\widehat{\delta}(\boldsymbol{X})}\bigg) +
\zeta(\widehat{P})\left(1 - \frac{\mathbb{P}_n[A]}{\widehat{P}[A = 1]}\right).
\end{align*}
For the LATE estimand $\chi(P)$, we instead follow \citet{takatsu2023drmliv} and take a ratio of one-step estimators for its numerator $\Gamma(P)$ and its denominator $\Delta(P)$, respectively:
\[\widehat{\chi} =
\frac{\mathbb{P}_n\left(\frac{2Z-1}
{\widehat{\pi}_Z(\boldsymbol{X})}\{Y -
\widehat{\mu}_Z(\boldsymbol{X})\} +
\widehat{\gamma}(\boldsymbol{X})\right)}
{\mathbb{P}_n\left(\frac{2Z-1}{\widehat{\pi}_Z(\boldsymbol{X})}\{A -
\widehat{\lambda}_Z(\boldsymbol{X})\} +
\widehat{\delta}(\boldsymbol{X})\right)}.\]
In other words, $\widehat{\chi} = \frac{\mathbb{P}_n(\phi_y(O; \widehat{P}))}{\mathbb{P}_n(\phi_a(O; \widehat{P}))}$, where $\phi_y(O; P) = \frac{2Z-1}{\pi_Z(\boldsymbol{X})}(Y - \mu_Z(\boldsymbol{X})) + \gamma(\boldsymbol{X})$, and $\phi_a(O; P) = \frac{2Z-1}{\pi_Z(\boldsymbol{X})}(A - \lambda_Z(\boldsymbol{X})) + \delta(\boldsymbol{X})$ are the uncentered influence functions of $\Gamma(P)$ and $\Delta(P)$, respectively.

Finally, we give the main result for this section, detailing the asymptotic behavior of the proposed estimators. The rates of convergence are related to the second-order remainder terms from the von Mises expansions in Lemma~\ref{lemma:VM-marginal}. As such, these convergence rates depend on products of errors in nuisance function estimation, permitting nominal $\sqrt{n}$-convergence and asymptotic normality for the parameters of interest, even under flexible specifications for the various component nuisance functions. For any function $f \in L_2(P)$, we will write $\lVert f \rVert^2 = \mathbb{E}_P(f(O)^2)$ for the squared $L_2(P)$ norm of $f$.

\begin{theorem}\label{thm:est-marginal}
    If $\lVert \xi(\, \cdot \, ; \widehat{P}) - \xi(\, \cdot \, ; P)\rVert = o_P(1)$, then
    \[\widehat{\xi} - \xi(P) = \mathbb{P}_n\left(\dot{\xi}(O; P)\right) + R_{\xi}(\widehat{P}, P) + o_P\left(n^{-1/2}\right).\]
    In particular, if $R_{\xi}(\widehat{P}, P) = o_P(n^{-1/2})$, then $\sqrt{n}(\widehat{\xi} - \xi(P)) \overset{d}{\to} \mathcal{N}(0, \mathrm{Var}_P(\dot{\xi}(O; P)))$, i.e., $\widehat{\xi}$ attains the nonparametric efficiency bound.
    
    Similarly, if $\lVert \zeta(\, \cdot \, ; \widehat{P}) - \zeta(\, \cdot \, ; P)\rVert = o_P(1)$, then
    \[\widehat{\zeta} - \zeta(P) = \mathbb{P}_n\left(\dot{\zeta}(O; P)\right) + R_{\zeta}(\widehat{P}, P) + o_P\left(n^{-1/2}\right),\]
    and if $R_{\zeta}(\widehat{P}, P) = o_P(n^{-1/2})$, then $\sqrt{n}(\widehat{\zeta} - \zeta(P)) \overset{d}{\to} \mathcal{N}(0, \mathrm{Var}_P(\dot{\zeta}(O; P)))$.
    
    If $\lVert \phi_y(\, \cdot \, ; \widehat{P}) - \phi_y(\, \cdot \, ; P)\rVert = o_P(1)$, $\lVert \phi_a(\, \cdot \, ; \widehat{P}) - \phi_a(\, \cdot \, ; P)\rVert = o_P(1)$, $\Delta(P) \neq 0$, and for some $\epsilon > 0$, $P[|\mathbb{P}_n(\phi_a(O; \widehat{P}))| \geq \epsilon] = 1$, then
    \[\widehat{\chi} - \chi(P) = \mathbb{P}_n\left(\dot{\chi}(O; P)\right) + O_P\left(R_y + R_a\right) + o_P\left(n^{-1/2}\right),\]
    where $R_y = \mathbb{E}_P\left((\widehat{\pi}_1 - \pi_1)\left\{\frac{\widehat{\mu}_1 - \mu_1}{\widehat{\pi}_1} + \frac{\widehat{\mu}_0 - \mu_0}{1 - \widehat{\pi}_1}\right\}\right)$, $R_a = \mathbb{E}_P\left((\widehat{\pi}_1 - \pi_1)\left\{\frac{\widehat{\lambda}_1 - \lambda_1}{\widehat{\pi}_1} + \frac{\widehat{\lambda}_0 - \lambda_0}{1 - \widehat{\pi}_1}\right\}\right)$. If further $R_a + R_y = o_P(n^{-1/2})$ then $\sqrt{n}(\widehat{\chi} - \chi(P)) \overset{d}{\to} \mathcal{N}(0, \mathrm{Var}_P(\dot{\chi}(O; P)))$.
\end{theorem}

\begin{remark}\label{rem:robust-marginal}
    Under the conditions of Theorem~\ref{thm:est-marginal}, and assuming $P[\epsilon \leq \widehat{\pi}_1 \leq 1 - \epsilon] = 1$,
    \[R_a + R_y = O_P\left(\left\lVert \widehat{\pi}_1 - \pi_1\right\rVert \left\{\left\lVert \widehat{\mu}_0 - \mu_0\right\rVert + \left\lVert \widehat{\mu}_1 - \mu_1\right\rVert + \left\lVert \widehat{\lambda}_0 - \lambda_0\right\rVert + \left\lVert \widehat{\lambda}_1 - \lambda_1\right\rVert\right\}\right).\]
    If also $P[\epsilon \leq \widehat{\delta}] = 1$, and $P[|\widehat{\gamma}| \leq 1/\epsilon] = 1$,
    \begin{align*}
        R_{\xi}(\widehat{P}, P) 
        &= O_P\bigg(\left\lVert \widehat{\pi}_1 - \pi_1\right\rVert \left\{\left\lVert \widehat{\mu}_0 - \mu_0\right\rVert + \left\lVert \widehat{\mu}_1 - \mu_1\right\rVert + \left\lVert \widehat{\lambda}_0 - \lambda_0\right\rVert + \left\lVert \widehat{\lambda}_1 - \lambda_1\right\rVert \right\}\\
        & \quad \quad \quad \quad  + \left\lVert\widehat{\delta} - \delta\right\rVert \left\lVert\frac{\widehat{\gamma}}{\widehat{\delta}} - \frac{\gamma}{\delta}\right\rVert\bigg).
    \end{align*}
    If further $\zeta(\widehat{P})\overset{P}{\to} \zeta(P)$, $\widehat{P}[A = 1] - P[A = 1] = O_P(n^{-1/2})$, and $P[|\rho| \leq 1/\epsilon] = 1$, then
    \begin{align*}
        R_{\zeta}(\widehat{P}, P) 
        &= O_P\bigg(\left\lVert \widehat{\pi}_1 - \pi_1\right\rVert \left\{\left\lVert \widehat{\mu}_0 - \mu_0\right\rVert + \left\lVert \widehat{\mu}_1 - \mu_1\right\rVert + \left\lVert \widehat{\lambda}_0 - \lambda_0\right\rVert + \left\lVert \widehat{\lambda}_1 - \lambda_1\right\rVert\right\} \\
        & \quad \quad \quad \quad + \left\{\left\lVert\widehat{\delta} - \delta\right\rVert + \left\lVert\widehat{\rho} - \rho\right\rVert\right\}\left\lVert\frac{\widehat{\gamma}}{\widehat{\delta}} - \frac{\gamma}{\delta}\right\rVert\bigg) + o_P\left(n^{-1/2}\right).
    \end{align*}
\end{remark}

Critically, Theorem~\ref{thm:est-marginal} shows that the proposed one-step estimators attain faster rates of convergence than the component nuisance function estimators. Specifically, in view of Remark~\ref{rem:robust-marginal}, convergence rates are on the order of products of nuisance function error rates. Moreover, in cases where these products are $o_P(n^{-1/2})$ (e.g., if all relevant component nuisance errors are $o_P(n^{-1/4})$), then the proposed estimators are $\sqrt{n}$-consistent, asymptotically normal and nonparametrically efficient; their asymptotic variances are equal to the variances of the influence functions of the corresponding functionals, which represent local asymptotic minimax lower bounds (see Corollary 2.6 of~\citet{vandervaart2002}). The necessary nuisance error rates for such asymptotic efficiency are attainable under structural conditions, such as smoothness or sparsity. Finally, we note that under the conditions that yield $\sqrt{n}$-consistency, Wald confidence intervals are asymptotically valid, e.g., a $(1-\alpha)$-level interval is given by
\[\widehat{\xi} \pm z_{1-\alpha/2} \sqrt{\frac{1}{n} \mathbb{P}_n\left(\dot{\xi}(O; \widehat{P})^2\right)},\]
where $z_{\alpha}$ is the $\alpha$-th quantile of the standard normal distribution---analogous confidence intervals can be constructed for $\zeta(P)$ and $\chi(P)$.


\subsection{Conditional Effects via Projection Parameters}\label{sec:est-project}
We now consider, for $v \in \mathrm{supp}(V)$, estimation of the conditional effects $\xi_P(v)$, $\zeta_P(v)$, $\chi_P(v)$, defined in~\eqref{eq:CATE}, \eqref{eq:CATT}, \eqref{eq:CLATE}, when $V$ contains continuous components and/or is high-dimensional. In such cases, these functionals are not, in general, pathwise differentiable. One might proceed by specifying parametric models, e.g., \citet{Tan:2006} proposed a two-stage regression, weighting, and doubly robust estimator of $\chi_P(v)$ for $V \equiv \boldsymbol{X}$ under specific parametric models for $\chi_P(V)$ and/or $\pi_1(\boldsymbol{X})$. Alternatively, as the parameters of interest involve structured combinations of regression functions, a more flexible approach would be to adapt the smooth functional-based methods from Sections~\ref{sec:est-marginal} to this infinite-dimensional setting---see the ``doubly robust learners'' of $\chi_P(v)$ in \citet{frauen2022, takatsu2023drmliv}, for instance. We leave such a direct approach for estimating $\xi_P(v)$ and $\zeta_P(v)$ for future work. In this paper, we will instead consider a working model approach, in line with \citet{neugebauer2007, rosenblum2010}, in which we target the projection of the true functions $\xi_P(\cdot)$, $\zeta_P(\cdot)$, $\chi_P(\cdot)$ onto working models. Namely, for $\theta \in \{\xi, \zeta, \chi\}$, define
\begin{equation}\label{eq:projection}
    \boldsymbol{\beta}_{\theta}(P) \coloneqq \mathrm{arg\,min}_{\boldsymbol{\beta} \in
        \mathbb{R}^q} \mathbb{E}_P\left[w(V) \left\{\theta_P(V) -
          m_{\theta}(V; \boldsymbol{\beta})\right\}^2\right],
\end{equation} where
    $\left\{m_{\theta}(V; \boldsymbol{\beta}) : \boldsymbol{\beta} \in \mathbb{R}^q\right\}$ is some working model (i.e., not necessarily assumed to contain the true function) for
    $\theta_P(V)$, and $w(V)$ a fixed
    non-negative weighting function. Note that while we adopt the $L_2(P)$-norm for projection, other choices are certainly possible. \citet{Abadie:2003a} took a similar approach for estimating the CLATE function $\chi_P(v)$.
    
    The proposed projection approach represents a middle ground between a fully parametric model-based specification and a direct model-free approach. From the definition~\eqref{eq:projection}, $\boldsymbol{\beta}_\theta(P)$ is the parameter corresponding to the best fitting (i.e., closest in distance) model of the form $m_{\theta}(V; \boldsymbol{\beta})$. Practically, the choice between considering $m_{\theta}(V; \boldsymbol{\beta})$ as a working model, or as a correctly specified model, amounts to a trade-off between bias and variance. In particular, if (a) the working model at hand happens to be correctly specified, then either the projection-based approach or one that assumes the working model is correct will yield valid inference, though potentially at some efficiency cost for the former; and if (b) the working model is not correctly specified, then the projection-based approach will yield valid inference for a well-defined summary measure, while the model-based approach may be biased. For further relevant commentary regarding such trade-offs, see the discussions in~\citet{kennedy2019, kennedy2023} where analogous projection parameters are proposed.

    As with the marginal parameters studied in Section~\ref{sec:est-marginal}, we will proceed with an influence function-based approach, beginning with the next result.
    
\begin{proposition}\label{prop:IF-projection}
The nonparametric efficient influence functions of the projection parameters $\boldsymbol{\beta}_{\theta}(P)$ are characterized up to proportionality as follows:

    \begin{align*}
      \dot{\boldsymbol{\beta}}_{\xi}(O; P)
      &\propto M_{\xi}(V; \boldsymbol{\beta}_{\xi}(P)) w(V) \bigg[
        \Psi_P(\boldsymbol{X}) - m_{\xi}(V; \boldsymbol{\beta}_{\xi}(P)) \\
      & \quad +
        \frac{2Z - 1}{\delta(\boldsymbol{X})\pi_Z(\boldsymbol{X})}\left\{
        Y - \mu_Z(\boldsymbol{X}) -
        \Psi_P(\boldsymbol{X})(A - \lambda_Z(\boldsymbol{X}))\right\}\bigg],
    \end{align*}

    \begin{align*}
      \dot{\boldsymbol{\beta}}_{\zeta}(O; P)
      &\propto M_{\zeta}(V; \boldsymbol{\beta}_{\zeta}(P)) w(V) \bigg[
        \zeta_P(V) - m_{\zeta}(V; \boldsymbol{\beta}_{\zeta}(P)) \\
      & \quad + \frac{1}{P[A = 1 \mid V]}\bigg\{A\left\{\Psi_P(\boldsymbol{X}) - \zeta_P(V)
        \right\}  \\
      & \quad \quad \quad + \frac{\rho(\boldsymbol{X})}{\delta(\boldsymbol{X})}
        \frac{2Z - 1}{\pi_Z(\boldsymbol{X})}\left\{Y - \mu_Z(\boldsymbol{X}) -
        \Psi_P(\boldsymbol{X})(A - \lambda_Z(\boldsymbol{X}))\right\}\bigg\}\bigg],
    \end{align*}
        \begin{align*}
      \dot{\boldsymbol{\beta}}_{\chi}(O; P)
      &\propto M_{\chi}(V; \boldsymbol{\beta}_{\chi}(P)) w(V) \bigg[
        \chi_P(V) - m_{\chi}(V; \boldsymbol{\beta}_{\chi}(P)) \\
      & \quad + \frac{1}{\mathbb{E}_P(\delta(\boldsymbol{X}) \mid V)}\bigg\{
        \gamma(\boldsymbol{X}) - \chi_P(V) \delta(\boldsymbol{X}) \\
      & \quad \quad \quad +
        \frac{2Z - 1}{\pi_Z(\boldsymbol{X})}\left\{ Y - \mu_Z(\boldsymbol{X}) -
        \chi_P(V)(A - \lambda_Z(\boldsymbol{X}))\right\}\bigg\}\bigg],
    \end{align*}
    where
    $M_{\theta}(V; \boldsymbol{\beta}) = \nabla_{\boldsymbol{\beta}}\,
    m_{\theta}(V; \boldsymbol{\beta})$, for $\theta \in \{\xi, \zeta, \chi\}$.
\end{proposition}
To estimate $\boldsymbol{\beta}_{\theta}(P)$, for $\theta = \xi, \zeta, \chi$, we will construct $\widehat{\boldsymbol{\beta}}_{\theta}$ by solving influence function-based estimating equations. Specifically, for fixed $\overline{P}$, let $\overline{\eta}_{\theta}$ collect the nuisance functions found in $\dot{\boldsymbol{\beta}}_{\theta}(O; \overline{P})$, i.e., 
\[\overline{\eta}_{\xi} \equiv (\overline{\pi}_1, \overline{\lambda}_0, \overline{\lambda}_1, \overline{\mu}_0, \overline{\mu}_1), \,\overline{\eta}_{\zeta} = (\overline{\eta}_\xi, \overline{P}[A = 1\mid V], \overline{\zeta}(V)), \, \overline{\eta}_{\chi} = (\overline{\eta}_\xi, \overline{E}(\delta \mid V),  \overline{\chi}(V)),\]
and define $\widehat{\boldsymbol{\beta}}_{\theta}$ such that
\[\widehat{\boldsymbol{\beta}}_{\theta} \text{ solves }
    \mathbb{P}_n\left[\phi_{\theta}(O; \boldsymbol{\beta},
      \widehat{\eta}_{\theta})\right] = o_P(n^{-1/2}),\] where
  $\phi_{\theta}(O; \boldsymbol{\beta}; \overline{\eta}_{\theta})$ is the scaled influence function
  $\dot{\boldsymbol{\beta}}_{\theta}(O; \overline{P})$ written out in Proposition \ref{prop:IF-projection},
  replacing any appearances of
  $\boldsymbol{\beta}_{\theta}(\overline{P})$ with
  $\boldsymbol{\beta}$. Given that we are solving for $\boldsymbol{\beta}$ that sets the empirical mean of the estimating functions to zero (up to $o_P(n^{-1/2})$ error), we are justified in ignoring proportionality constants in Proposition~\ref{prop:IF-projection}.

  Now, we characterize the asymptotic behavior of the proposed estimators, giving the analog of Theorem~\ref{thm:est-marginal} for the estimators of the parameters of interest in this section. 

  \begin{theorem}\label{thm:est-projection}
    For $\theta \in \{\xi, \zeta, \chi\}$, suppose the following hold:
    \begin{enumerate}[label=(\arabic*)]
    \item The function class
      $\left\{\phi_{\theta}(o; \boldsymbol{\beta}, \overline{\eta}_{\theta}):
        \boldsymbol{\beta} \in \mathbb{R}^q\right\}$ is Donsker (in
      $\boldsymbol{\beta}$) for each fixed $\overline{\eta}_{\theta}$,
    \item
      $\widehat{\boldsymbol{\beta}}_{\theta} - \boldsymbol{\beta}_{\theta}(P) = o_P(1)$
      and
      $\left\lVert \widehat{\eta}_{\theta} - \eta_{\theta}\right\rVert
      = o_P(1)$, with $\left\lVert \cdot \right\rVert$ interpreted component-wise,
    \item The map
      $\boldsymbol{\beta} \mapsto \mathbb{E}_P\left(\phi_{\theta}(O;
        \boldsymbol{\beta}, \overline{\eta}_{\theta})\right)$ is differentiable
      at $\boldsymbol{\beta}_{\theta}(P)$, uniformly in
      $\overline{\eta}_{\theta}$, with invertible derivative matrix
      $V_{\theta}(\boldsymbol{\beta}_{\theta}(P), \overline{\eta}_{\theta}) =
      \left. \nabla_{\boldsymbol{\beta}}\mathbb{E}_P\left(\phi_{\theta}(O;
          \boldsymbol{\beta}, \overline{\eta}_{\theta})\right)
      \right|_{\boldsymbol{\beta} = \boldsymbol{\beta}_{\theta}(P)}$,
      such that
      \[V_{\theta}(\boldsymbol{\beta}_{\theta}(P), \widehat{\eta}_{\theta})
      \overset{P}{\to} V_{\theta}(\boldsymbol{\beta}_{\theta}(P), \eta_{\theta}).\]
    \end{enumerate}
    Then
    \[\widehat{\boldsymbol{\beta}}_{\theta} -
      \boldsymbol{\beta}_{\theta}(P) =
      -V_{\theta}(\boldsymbol{\beta}_{\theta}(P), \eta_{\theta})^{-1}
      \mathbb{P}_n\left(\phi_{\theta}(O;
        \boldsymbol{\beta}_{\theta}(P), \eta_{\theta})\right) + O_P(R_{n, \theta})
      + o_P\left(n^{-1/2}\right),\] where
    $R_{n, \theta} = \mathbb{E}_P\left(\phi_{\theta}(O;
      \boldsymbol{\beta}_{\theta}(P), \widehat{\eta}_{\theta}) - \phi_{\theta}(O;
      \boldsymbol{\beta}_{\theta}(P), \eta_{\theta})\right)$ are given by
    \begin{align*}
      R_{n, \xi}
      &= \mathbb{E}_P\bigg(M_{\xi}(V; \boldsymbol{\beta}_{\xi}(P))w(V)
        \frac{1}{\widehat{\delta}}
        \bigg\{\left(\frac{\widehat{\gamma}}{\widehat{\delta}}-
        \frac{\gamma}{\delta}\right)
        \left(\widehat{\delta} - \delta\right) \\
      & \quad \quad + \left(\widehat{\pi}_1 - \pi_1\right)
        \bigg(\frac{\widehat{\mu}_1 - \mu_1 -
        \frac{\widehat{\gamma}}{\widehat{\delta}}(\widehat{\lambda}_1 -
        \lambda_1)}{\widehat{\pi}_1} + \frac{\widehat{\mu}_0 - \mu_0 -
        \frac{\widehat{\gamma}}{\widehat{\delta}}(\widehat{\lambda}_0 -
        \lambda_0)}{1 - \widehat{\pi}_1}\bigg)
        \bigg\}\bigg),
    \end{align*}
    \begin{align*}
      R_{n, \zeta}
      &= \mathbb{E}_P\bigg(M_{\zeta}(V; \boldsymbol{\beta}_{\zeta}(P))w(V)
        \frac{1}{\widehat{P}[A = 1 \mid V]}
        \bigg\{\left(\frac{\widehat{\gamma}}{\widehat{\delta}} -
        \frac{\gamma}{\delta}\right)
        \left((\rho - \widehat{\rho}) + \frac{\widehat{\rho}}{\widehat{\delta}}(\widehat{\delta} - \delta)\right) \\
        & \quad \quad +\left(\widehat{\zeta}(V) - \zeta_P(V)\right)\left(\widehat{P}[A = 1 \mid V] - P[A = 1 \mid V]\right) \\
      & \quad \quad + \frac{\widehat{\rho}}{\widehat{\delta}}\left(\widehat{\pi}_1 - \pi_1\right)
        \bigg(\frac{\widehat{\mu}_1 - \mu_1 -
        \frac{\widehat{\gamma}}{\widehat{\delta}}(\widehat{\lambda}_1 -
        \lambda_1)}{\widehat{\pi}_1} + \frac{\widehat{\mu}_0 - \mu_0 -
        \frac{\widehat{\gamma}}{\widehat{\delta}}(\widehat{\lambda}_0 -
        \lambda_0)}{1 - \widehat{\pi}_1}\bigg)
        \bigg\}\bigg),
    \end{align*}
        \begin{align*}
      R_{n, \chi}
      &= \mathbb{E}_P\bigg(M_{\chi}(V; \boldsymbol{\beta}_{\chi}(P))w(V)
        \frac{1}{\widehat{\mathbb{E}}(\delta \mid V)}
        \bigg\{\left(\widehat{\chi}(V) - \chi_P(V)\right)
        \left(\widehat{\mathbb{E}}(\delta \mid V) -
        \mathbb{E}_P(\delta \mid V)
        \right) \\
      & \quad \quad + \left(\widehat{\pi}_1 - \pi_1\right)
        \bigg(\frac{\widehat{\mu}_1 - \mu_1 -
        \widehat{\chi}(V)(\widehat{\lambda}_1 -
        \lambda_1)}{\widehat{\pi}_1} + \frac{\widehat{\mu}_0 - \mu_0 -
        \widehat{\chi}(V)(\widehat{\lambda}_0 -
        \lambda_0)}{1 - \widehat{\pi}_1}\bigg)
        \bigg\}\bigg).
    \end{align*}
  \end{theorem}
In view of Theorem \ref{thm:est-projection}, as for the marginal parameters considered in Section~\ref{sec:est-marginal}, the proposed estimators converge (under weak conditions) at rates proportional to the product of nuisance estimation errors. The three conditions used in the result deserve some discussion. The first condition requires that the influence function is a relatively simple function of $\boldsymbol{\beta}$, though note that we do not impose any restriction on the complexity with respect to $\overline{\eta}_{\theta}$. This will hold, for instance, if the working model is a generalized linear model. The second condition states that $(\widehat{\boldsymbol{\beta}}, \widehat{\eta}_{\theta})$ must be consistent at any rate of convergence, which we also view as a relatively weak requirement. Finally, the third condition calls for a degree of smoothness in the influence function with respect to $\boldsymbol{\beta}$, in order for a delta method-type argument to be invoked.

Similar to the case of the marginal parameters, Theorem~\ref{thm:est-projection} shows that when the second order remainder terms converge to zero fast enough (i.e., when $R_{n, \xi}, R_{n, \zeta}, R_{n, \chi} = o_P(n^{-1/2})$), the estimator is $\sqrt{n}$-consistent, asymptotically normal, and attains the nonparametric efficiency bound. Further, in this case, simple Wald confidence intervals are asymptotically valid. To be concrete, note that
\small
    \[V_{\theta}(\boldsymbol{\beta}_{\theta}(P), \eta_{\theta}) =
      \mathbb{E}_P\left(w(V)\left\{\Lambda_{\theta}(V;
          \boldsymbol{\beta}_{\theta}(P))\left[\theta_P(V) - m_{
              \theta}(V; \boldsymbol{\beta}_{\theta}(P))\right] -
          M_{\theta}^{\otimes 2}(V;
          \boldsymbol{\beta}_{\theta}(P))\right\}\right),\]
\normalsize
where
  $\Lambda_{\theta}(V; \boldsymbol{\beta}) =
  \nabla_{\boldsymbol{\beta}}^2 m_{ \theta}(V;
  \boldsymbol{\beta})$ and $\boldsymbol{u}^{\otimes 2} = \boldsymbol{u} \boldsymbol{u}^T$ for any $\boldsymbol{u} \in \mathbb{R}^q$. Thus, defining
    \[\widehat{V}_{\theta} =
      \mathbb{P}_n\left(w(V)\left\{\Lambda_{\theta}(V;
          \widehat{\boldsymbol{\beta}}_{\theta})\left[\widehat{\theta}(V)
            - m_{ \theta}(V; \widehat{\boldsymbol{\beta}}_{\theta})\right] -
          M_{\theta}^{\otimes 2}(V;
          \widehat{\boldsymbol{\beta}}_{\theta})\right\}\right),\]
a $(1 - \alpha)$-level confidence interval is given by $\widehat{\boldsymbol{\beta}}_{\theta} \pm z_{1 - \alpha / 2} \widehat{V}_{\theta}^{-1}\sqrt{ \mathbb{P}_n\left(\phi_{\theta}(O; \widehat{\boldsymbol{\beta}}_{\theta}, \widehat{\eta}_{\theta})^2\right) / n}$.

\section{Applications}

Next, we present two different IV applications. The first application is from the political science literature, and the IV arises due to noncompliance in a randomized experiment. The second application is from the comparative effectiveness research (CER) literature. In this application, the IV is naturally occurring. For each application, we highlight relevant differences between the possible marginal IV estimands.

\subsection{Markets}
In experiments with non-compliance, randomized trial arm assignment is a natural instrumental variable. \citet{margalit2021markets} conducted a field experiment to test the effect of participation in financial markets (through investing) on political views. In particular, a sample of the English adult population in 2016 was randomly assigned to an ``asset treatment'' ($Z = 1$) in which participants made weekly investments in a portfolio with real financial assets, using either real or ``fantasy money'', or assets indexed to team performance in American baseball---or to a control condition ($Z = 0$) which involved no investment decisions. Some participants that were assigned to the investment arm did not actually invest, and thus crossed over to the control condition.  However, subjects assigned to the control condition were unable to access the treatment condition. As such, in this example, one-sided non-compliance (see Section~\ref{sec:monotonicity}) holds by design in this study; such that if $A$ is an indicator of active investment in line with the study protocol, $A = 1$ only if $Z = 1$.  The outcome ($Y$) in this study was a socioeconomic values score, computed on the basis of survey responses related to personal responsibility, economic fairness, inequality, and redistribution. Higher values on this score indicate more traditionally right-wing views. In total, $n = 2,223$ participants completed all pre- and post-treatment surveys, and were included in this analysis. The baseline covariates $\boldsymbol{X}$ used in all models included baseline socioeconomic value score, an indicator of having ever invested prior to the study, party voted for in 2015, income, age, gender, educational attainment, and a risk tolerance score.

For our analysis, we estimate the marginl LATE, ATE, and ATT estimands as laid out in Section~\ref{sec:est-marginal}, using ten-fold cross-fitting. Here, we fit each of the nuisance functions $\widehat{\lambda}_0, \widehat{\lambda}_1, \widehat{\mu}_0, \widehat{\mu}_1$, using an ensemble \texttt{SuperLearner} \citep{polley2019}, with generalized linear model, regression tree, random forest, and polynomial spline regression as the library of learners. On the other hand, we invoke randomization and use the marginal empirical distribution of $Z$ for $\widehat{\pi}_0, \widehat{\pi}_1$. The estimated LATE was 0.125 (95\% CI: -0.007, 0.257), the estimated ATE was 0.121 (95\% CI: -0.011, 0.252), and the estimated ATT was 0.122 (95\% CI: -0.010, 0.254). Assuming sufficient identification conditions hold to give each of these estimands a causal interpretation, our results indicate that there is not much difference between the effect of participation in markets on political views, in the whole population versus in the treated or among compliers. In this study, it is easiest to justify the assumptions underpinning the LATE (via Proposition~\ref{prop:LATE}) and the ATT (via Proposition~\ref{prop:one-sided}), since randomization guarantees unconfoundedness of the IV, and monotonicity in the form of one-sided non-compliance holds by design. The results for the ATE depend on the whether the measured covariates are able to capture effect modification, as required for any of the various identification results.

\subsection{Tendency to Operate}

In CER, researchers often use a physician's preference for a specific treatment as an IV \citep{brookhart2007preference,keeleegsiv2018}. For this type of IV, each physician preference serves as a ``nudge" toward a specific mode of care; where the analyst assumes that the physician's preference has no direct effect on patient outcomes. IVs of this type were first used for applications in pharmacoepidemiology \citep{brookhart2006evaluating}. More recently, IVs of this type have been used to study the effectiveness of operative care for a range of emergency conditions \citep{keeleegsiv2018,kaufman2023operative,takatsu2023drmliv,hutchings2022effectiveness}. Here, we re-analyze data from \citet{takatsu2023drmliv} who use an IV of this type to study the effectiveness of surgery for cholecystitis patients.

In this application, the outcome $Y$ is a binary indicator for an adverse event (a complication or prolonged length of stay). The treatment $A$ is binary, where $A=1$ indicates that a patient received operative care after their emergency admission, while $A = 0$ indicates that a patient received non-operative care.  The instrument $Z$ is an indicator that takes the value of $1$ when a surgeon has a high tendency to operate (TTO) and 0 otherwise. We measure TTO using the percentage of times a surgeon operates when presented with an emergency surgery condition, and record $Z=1$ when a surgeon has TTO value higher than the sample median. In other words, an IV value of 1 indicates that a patient has been assigned to a high TTO surgeon and thus are more likely to receive operative care. Finally, baseline covariates $\boldsymbol{X}$ include race, sex, age, insurance type, baseline frailty measures, comorbidities, and hospital. In this data, there are 123,420 patients with 78,337 assigned to a high TTO surgeon and 45,083 assigned to a low TTO surgeon. 

The difference between the various IV estimands is of particular interest in this study. Most studies of this type invoke some form of a monotonicity assumption~\citep{keeleegsiv2018,kaufman2023operative,takatsu2023drmliv,hutchings2022effectiveness}. In this context, the LATE represents the effect of surgery among those patients who receive surgery because they receive care from a high TTO surgeon. As such, the LATE provides information about the effect of surgery for those patients that are marginal for treatment. However, it seems unlikely that the LATE coincides with a more common target causal estimand like the ATE. Specifically, results in past analyses have found that the always-taker population tends to be comprised of the healthiest patients, while never-takers tend to be the most frail patients \citep{keeleegsiv2018,kaufman2023operative,takatsu2023drmliv}. Compliers tend to fall between aways-takers and never-takers in terms of health at baseline. The ATE, by contrast, measures the average difference in outcomes when all individuals in the study population are assigned to the surgery group versus when all individuals are assigned to the non-operative management group. Finally, the ATT measures the average difference in outcomes among those individuals in the population that actually had surgery. We might expect that both of these estimands may differ markedly from the LATE, which only represents the effect among compliers. That is, the effect of surgery on never-takers and always-takers may be quite different than it is for the compliers.

In our analysis, we use DRML IV methods to estimate the LATE, ATE, and ATT estimands, invoking any of the sufficient identifying assumptions laid out in Section~\ref{sec:identification}. We have a large sample size, so we only use two splits of the data for cross-validation. We use an ensemble of methods for the estimation of the nuisance functions. Specifically, we estimate nuisance function quantities using an ensemble \texttt{SuperLearner} \citep{polley2019} with four different learners: a generalized linear model, a generalized additive model, a random forest, and the lasso. First, we report results for the LATE estimand. Here, we find that surgery reduces the risk of an adverse event by an estimated 3.8\% (95\% CI: -0.031, -0.045) among compliers. Next, for the ATE estimand, we find that surgery reduces the risk of an adverse event by as estimated 4.9\% (95\% CI: -0.044, -0.055) over the population of cholecystitis patients.  As such, the estimated effect for the ATE estimand is notably larger in magnitude than that for the LATE, though the 95\% confidence intervals do narrowly overlap. Lastly, the estimated ATT corresponds to a reduction in risk of 5.6\% (95\% CI: -0.050, -0.061) among treated patients. Absent concerns about sampling variability, we find that the more local estimand tends to understate the effectiveness of operative treatment for cholecystitis patients. In this application, given that the IV is naturally occurring, none of the assumptions are satisfied by design. \citet{keeleegsiv2018} and \citet{takatsu2023drmliv} provide detailed information about the credibility of the assumptions for TTO as an IV. Again, the results for the ATE depend on the whether the measured covariates are rich enough to be the full set of effect modifiers.

\section{Discussion}

The focus on the concept of the \textit{estimand} is a foundational component of modern causal analysis. Often in the literature, it is noted that the choice of estimand should be based on a scientific judgement and should reflect the specific research question. However, it is often the case that there is some interplay between the estimand and the plausibility of the identification assumptions. Nowhere is this more true than in an IV analysis. This may account for the popularity of the LATE estimand. That is, in many applications, the monotonicity assumption is highly plausible---the likelihood of defiers being present is unlikely, or may even be guaranteed by the study design. 

One goal of this paper is to fully outline that both the ATE (or CATE) and ATT (or CATT) estimands can also be targeted in an IV analysis. In our reading of the literature, these two estimands are rarely the target of interest in applied IV studies. Here, we reviewed the panoply of assumptions that can be invoked to target these two estimands. In the decades of IV methodology research, the assumptions needed to identify the CATE and ATE have evolved considerably. Originally, targeting the ATE required invoking often implausibly simple structural equations. Recent work has outlined much weaker assumptions that focus on conditioning on key effect modifiers. In addition, we also outlined a robust and efficient estimation framework for targeting all these different estimands in real data. Here, we focused on modern methods based on nonparametric efficiency theory and influence functions. These methods are doubly-robust and allow analysts to use a variety of flexible estimation methods to reduce the likelihood of bias from model misspecification. We demonstrated these estimation methods and compared these estimands in two very different empirical applications. This should serve as a useful reference to analysts, so they can better use IV methods in applied research.

\section*{Acknowledgements}
AWL would like to thank Professor Andrea Rotnitzky for her valuable comments on the manuscript. Professor Rotnitzky's course notes on instrumental variables strongly influenced our review of identification schemes.

\section*{Funding}
Research in this article was supported by the Patient-Centered Outcomes Research Institute (PCORI Awards ME-2021C1-22355) and by the National Library of Medicine, \#1R01LM013361-01A1. All statements in this report, including its findings and conclusions, are solely those of the authors and do not necessarily represent the views of PCORI or its Methodology Committee.

\clearpage

\section*{References}
\vspace{-1cm}
\bibliographystyle{asa}
\bibliography{bibliography.bib}

\pagebreak

\begin{appendices}

\section{Proofs of Results in Section~\ref{sec:identification}}

\begin{proof}[Proof of Proposition~\ref{prop:homogeneity-ident}]

    Letting $a \in \mathcal{A}$ be arbitrary, \eqref{eq:strict-homogeneity} implies $Y - Y(a) = r(\boldsymbol{X})(A - a)$, so that $Y(a) = Y - r(\boldsymbol{X})(A - a)$. Since $Z \independent Y(a) \mid \boldsymbol{X}$, 
    \begin{align*}
        0 = \mathrm{Cov}(Z, Y(a) \mid \boldsymbol{X}) &= \mathrm{Cov}(Z, Y - r(\boldsymbol{X})(A - a) \mid \boldsymbol{X})  \\
        &= \mathrm{Cov}_P(Z, Y \mid \boldsymbol{X}) - r(\boldsymbol{X})\mathrm{Cov}_P(Z, A \mid \boldsymbol{X}),
    \end{align*}
    and rearranging yields $r(\boldsymbol{X}) = \Psi_P(\boldsymbol{X})$. We may safely divide by $\mathrm{Cov}_P(Z, A \mid \boldsymbol{X})$ by Assumption~\ref{ass:relevance}.
\end{proof}

\begin{proof}[Proof of Lemma~\ref{lemma:binary}]
    For $W$ taking the place of either $A$ or $Y$,
    \begin{align*}
    & \mathrm{Cov}(Z, W \mid \boldsymbol{X}) \\
    &= \mathbb{E}(W\{Z - \mathbb{E}(Z \mid \boldsymbol{X}) \} \mid \boldsymbol{X}) \\
    &= \mathbb{E}_P(W \mid \boldsymbol{X}, Z = 1) \pi_1(\boldsymbol{X}) \pi_0(\boldsymbol{X}) + \mathbb{E}_P(W \mid \boldsymbol{X}, Z = 0) \pi_0(\boldsymbol{X})\left\{-\pi_1(\boldsymbol{X})\right\} \\
    &= \pi_1(\boldsymbol{X}) \pi_0(\boldsymbol{X}) \left\{\mathbb{E}_P(W \mid \boldsymbol{X}, Z = 1) - \mathbb{E}_P(W \mid \boldsymbol{X}, Z = 0)\right\}.
    \end{align*}
    The result then follows by cancellation from the definition of $\Psi_P(\boldsymbol{X})$.
\end{proof}

\begin{proof}[Proof of Proposition~\ref{prop:nonconstant-ident}]
    As in the proof of Proposition \ref{prop:homogeneity-ident}, with $a \in \mathcal{A}$ arbitrary, Assumption~\ref{ass:UC}(a) implies
    \begin{align*}
        0 = \mathrm{Cov}(Z, Y(a) \mid \boldsymbol{X}) &= \mathrm{Cov}(Z, Y - h(\boldsymbol{X}, \epsilon_Y)(A - a) \mid \boldsymbol{X})  \\
        &= \mathrm{Cov}_P(Z, Y \mid \boldsymbol{X}) - \mathrm{Cov}_P(Z, h(\boldsymbol{X}, \epsilon_Y)A \mid \boldsymbol{X}),
    \end{align*}
    but note that
    \begin{align*}
        \mathrm{Cov}_P(Z, h(\boldsymbol{X}, \epsilon_Y)A \mid \boldsymbol{X}) 
        &= \mathbb{E}(\{Z - \mathbb{E}_P(Z \mid \boldsymbol{X})\}h(\boldsymbol{X}, \epsilon_Y)A) \mid \boldsymbol{X}) \\
        &= \mathbb{E}(\{Z - \mathbb{E}_P(Z \mid \boldsymbol{X})\}\mathbb{E}(h(\boldsymbol{X}, \epsilon_Y) \mid \boldsymbol{X}, Z, A)A \mid \boldsymbol{X}) \\
        &= \mathbb{E}(h(\boldsymbol{X}, \epsilon_Y) \mid \boldsymbol{X})\mathrm{Cov}_P(Z, A \mid \boldsymbol{X}),
    \end{align*}
    where we conditioned on $(\boldsymbol{X}, Z, A)$ in the second equality, and used \eqref{eq:nonconstant-ass} in the third equality. Thus, the result is obtained by rearranging.
\end{proof}

\begin{proof}[Proof of Propositions~\ref{prop:wang-outcome-ident}, \ref{prop:wang-exposure-ident}, and \ref{prop:cui-cov-ident}]

Let $\tau(\boldsymbol{X}) = \mathbb{E}(Y(a = 1) - Y(a = 0) \mid \boldsymbol{X})$ be the CATE. It is sufficient to show that
  \begin{equation}\label{eq:ind}
    \mathbb{E}(Y - \tau(\boldsymbol{X}) A \mid \boldsymbol{X}, Z)
    \text{ does not depend on } Z.
  \end{equation}
  To see why, notice that \eqref{eq:ind} implies that
  \[\mathbb{E}(Y - \tau(\boldsymbol{X}) A \mid \boldsymbol{X}, Z) =
  \mathbb{E}\left\{\mathbb{E}(Y - \tau(\boldsymbol{X}) A \mid Z, \boldsymbol{X}) \mid
    \boldsymbol{X}\right\} = \mathbb{E}(Y - \tau(\boldsymbol{X}) A \mid \boldsymbol{X}),\] hence
  \begin{align*}
    \mathbb{E}(\{Z - \mathbb{E}_P(Z \mid  \boldsymbol{X})\}(Y - \tau(\boldsymbol{X}) A) \mid \boldsymbol{X})
    &= \mathbb{E}(\{Z -
      \mathbb{E}(Z \mid \boldsymbol{X})\}
      \mathbb{E}((Y - \tau(\boldsymbol{X}) A) \mid \boldsymbol{X}, Z) \mid \boldsymbol{X}) \\
    &=
      \mathbb{E}(Z - \mathbb{E}(Z \mid \boldsymbol{X}) \mid \boldsymbol{X}) \cdot
      \mathbb{E}(Y - \tau(\boldsymbol{X}) A \mid \boldsymbol{X}) = 0,
  \end{align*}
  meaning that
  $\mathrm{Cov}(Z, Y - \tau(\boldsymbol{X}) A \mid \boldsymbol{X}) = 0$. Using
  bilinearity of conditional covariance and solving for $\tau(\boldsymbol{X})$,
  this implies that
  $\tau(\boldsymbol{X}) = \frac{\mathrm{Cov}_P(Z, Y \mid \boldsymbol{X})}{\mathrm{Cov}_P(Z,
    A \mid \boldsymbol{X})} = \Psi_P(\boldsymbol{X})$.

\vspace{3mm}
 It remains to prove \eqref{eq:ind} under the conditions of Propositions~\ref{prop:wang-outcome-ident}, \ref{prop:wang-exposure-ident}, and \ref{prop:cui-cov-ident}. Since $A$ is binary, consistency yields $Y - \tau(\boldsymbol{X})A = A(Y(a = 1) - Y(a = 0) - \tau(\boldsymbol{X})) + Y(a = 0)$. Assumption~\ref{ass:UC}(a) implies that $\mathbb{E}(Y(a = 0) \mid \boldsymbol{X}, Z) = \mathbb{E}(Y(a = 0) \mid \boldsymbol{X})$; recall also that Assumption~\ref{ass:UC}(a) is implied by Assumptions~\ref{ass:UC}(c)\&(d). To establish~\eqref{eq:ind}, we now need only show that $\mathbb{E}(A\{Y(a = 1) - Y(a = 0) - \tau(\boldsymbol{X})\} \mid \boldsymbol{X}, Z)$ does not depend on $Z$. Observe that
 \begin{align*}
     &\mathbb{E}(A\{Y(a = 1) - Y(a = 0) - \tau(\boldsymbol{X})\} \mid \boldsymbol{X}, Z) \\
     &= \mathbb{E}\left[\mathbb{E}(A\{Y(a = 1) - Y(a = 0) - \tau(\boldsymbol{X})\} \mid \boldsymbol{X}, Z, U) \mid \boldsymbol{X}, Z\right] \\
     &= \mathbb{E}\left[\mathbb{E}\{A \mid \boldsymbol{X}, Z, U\} \mathbb{E}\{Y(a = 1) - Y(a = 0) - \tau(\boldsymbol{X}) \mid \boldsymbol{X}, Z, U \} \mid \boldsymbol{X}, Z\right] \\
     &= \mathbb{E}\left[\mathbb{E}\{A \mid \boldsymbol{X}, Z, U\} \mathbb{E}\{Y(a = 1) - Y(a = 0) - \tau(\boldsymbol{X}) \mid \boldsymbol{X}, U \} \mid \boldsymbol{X}, Z\right] \\
     &= \mathbb{E}\left[\left\{Z \widetilde{\delta}(\boldsymbol{X}, U) + \mathbb{E}(A \mid \boldsymbol{X}, Z = 0, U)\right\}\left\{\widetilde{\tau}(\boldsymbol{X}, U) - \tau(\boldsymbol{X})\right\} \mid \boldsymbol{X}, Z\right] \\
     &= Z \cdot \mathbb{E}\left[\widetilde{\delta}(\boldsymbol{X}, U)\left\{\widetilde{\tau}(\boldsymbol{X}, U) - \tau(\boldsymbol{X})\right\} \mid \boldsymbol{X}, Z\right] \\
     & \quad \quad + \mathbb{E}\left[\mathbb{E}(A \mid \boldsymbol{X}, Z = 0, U)\left\{\widetilde{\tau}(\boldsymbol{X}, U) - \tau(\boldsymbol{X})\right\} \mid \boldsymbol{X}, Z\right],
 \end{align*}
 where the first equality follows by the tower law, the second and third from Assumption~\ref{ass:UC}(c), the fourth by definition of $\widetilde{\delta}(\boldsymbol{X}, U)$ and $\widetilde{\tau}(\boldsymbol{X}, U)$, and the last by rearranging. We see that the result of Proposition~\ref{prop:wang-outcome-ident} is obtained immediately from the above expression, since~\eqref{eq:wang-outcome} states that $\widetilde{\tau}(\boldsymbol{X}, U) = \tau(\boldsymbol{X})$. For Propositions~\ref{prop:wang-exposure-ident} and~\ref{prop:cui-cov-ident}, we further assume $Z \independent U \mid \boldsymbol{X}$ (i.e., Assumption~\ref{ass:UC}(d)), which simplifies the above equality:
 \begin{align*}
      &\mathbb{E}(A\{Y(a = 1) - Y(a = 0) - \tau(\boldsymbol{X})\} \mid \boldsymbol{X}, Z) \\
      &= Z \cdot \mathrm{Cov}\left(\widetilde{\delta}(\boldsymbol{X}, U), \widetilde{\tau}(\boldsymbol{X}, U) \mid \boldsymbol{X}\right) + \mathbb{E}\left[\mathbb{E}(A \mid \boldsymbol{X}, Z = 0, U)\left\{\widetilde{\tau}(\boldsymbol{X}, U) - \tau(\boldsymbol{X})\right\} \mid \boldsymbol{X} \right].
 \end{align*}
 Clearly, when $\mathrm{Cov}\left(\widetilde{\delta}(\boldsymbol{X}, U), \widetilde{\tau}(\boldsymbol{X}, U) \mid \boldsymbol{X}\right) = 0$ (i.e., when equation~\eqref{eq:cui-cov} assumed in Proposition~\ref{prop:cui-cov-ident} holds), the right-hand side of the above display will not depend on $Z$. Moreover, this null conditional covariance is guaranteed by \eqref{eq:wang-exposure}, assumed in Proposition~\ref{prop:wang-exposure-ident}.
\end{proof}

\begin{proof}[Proof of Proposition~\ref{prop:SMM-ident}]
    Note that, $\mathbb{E}(Y - r(\boldsymbol{X})A \mid \boldsymbol{X}, Z, A) = \mathbb{E}(Y(a = 0) \mid \boldsymbol{X}, Z, A)$, by rearranging model \eqref{eq:SMM}. This implies that
    \[\mathbb{E}(Y - r(\boldsymbol{X})A \mid \boldsymbol{X}, Z) = \mathbb{E}(Y(a = 0) \mid \boldsymbol{X}, Z), \]
    integrating out $A$. But by Assumption~\ref{ass:UC}(a), $\mathbb{E}(Y(a = 0) \mid \boldsymbol{X}, Z) = \mathbb{E}(Y(a = 0) \mid \boldsymbol{X})$, which does not depend on $Z$, i.e., $\mathbb{E}(Y - r(\boldsymbol{X})A \mid \boldsymbol{X}, Z) = \mathbb{E}(Y - r(\boldsymbol{X})A \mid \boldsymbol{X})$. Thus, we have
    \begin{align*}
        \mathrm{Cov}(Z, Y - r(\boldsymbol{X})A \mid \boldsymbol{X}) 
        &=
        \mathbb{E}\left(\left\{Z-\mathbb{E}_P(Z \mid \boldsymbol{X})\right\} \left\{Y - r(\boldsymbol{X})A\right\} \mid \boldsymbol{X}\right) \\
        &=\mathbb{E}\left(\left\{Z-\mathbb{E}_P(Z \mid \boldsymbol{X})\right\} \mathbb{E}\left(Y - r(\boldsymbol{X})A \mid \boldsymbol{X}, Z\right) \mid \boldsymbol{X}\right) \\
        &= \mathbb{E}(Y - r(\boldsymbol{X})A \mid \boldsymbol{X})\mathbb{E}_P\left(Z-\mathbb{E}_P(Z \mid \boldsymbol{X})\mid \boldsymbol{X}\right) \\
        &= 0,
    \end{align*}
    by conditioning on $(\boldsymbol{X}, Z)$. Thus, rearranging, $r(\boldsymbol{X}) = \Psi_P(\boldsymbol{X})$.
\end{proof}

\begin{proof}[Proof of Proposition~\ref{prop:hernan}]
    It suffices, by Proposition~\ref{prop:SMM-ident}, to show that $r(\boldsymbol{X}) = s(\boldsymbol{X})$. To ease notation, let $D = Y(a = 1) - Y(a = 0)$, and note that
    \[r(\boldsymbol{X}) = \mathbb{E}(D \mid \boldsymbol{X}, Z = 0, A = 1) = \mathbb{E}(D \mid \boldsymbol{X}, Z = 1, A = 1),\]
    and similarly,
    \[s(\boldsymbol{X}) = \mathbb{E}(D \mid \boldsymbol{X}, Z = 0, A = 0) = \mathbb{E}(D \mid \boldsymbol{X}, Z = 1, A = 0).\]
    By Assumption~\ref{ass:UC}(a), we have
    \[\mathbb{E}(D \mid \boldsymbol{X}, Z = 0) = \mathbb{E}(D \mid \boldsymbol{X}, Z = 1).\]
    Thus, by the law of total probability,
    \begin{align*} 
    & P[A = 0 \mid \boldsymbol{X}, Z = 0] s(\boldsymbol{X}) + 
    P[A = 1 \mid \boldsymbol{X}, Z = 0] r(\boldsymbol{X}) \\
    =& P[A = 0 \mid \boldsymbol{X}, Z = 1] s(\boldsymbol{X}) + 
    P[A = 1 \mid \boldsymbol{X}, Z = 1] r(\boldsymbol{X}).
    \end{align*}
    Rearranging, we obtain,
    \[\{r(\boldsymbol{X}) - s(\boldsymbol{X})\}\{P[A = 1 \mid \boldsymbol{X}, Z = 1] - P[A = 1 \mid \boldsymbol{X}, Z = 0]\} = 0,\]
    which by Assumption~\ref{ass:relevance} yields $r(\boldsymbol{X}) = s(\boldsymbol{X})$.
\end{proof}

\begin{proof}[Proof of Proposition~\ref{prop:LATE}]
    Observe that
    \begin{align*}
        \mu_1(\boldsymbol{X}) - \mu_0(\boldsymbol{X})
        &= \mathbb{E}_P(Y(z = 1) - Y(z = 0) \mid \boldsymbol{X}) \\
        &= \mathbb{E}_P(Y(z = 1, A(z = 1)) - Y(z = 0, A(z = 0)) \mid \boldsymbol{X}) \\
        &= \mathbb{E}_P(Y(A(z = 1)) - Y(A(z = 0)) \mid \boldsymbol{X}) \\
        &= \mathbb{E}_P\left(\left\{Y(a = 1) - Y(a = 0)\right\} \left\{A(z = 1) - A(z = 0)\right\} \mid \boldsymbol{X}\right),
    \end{align*}
    where the first equality follows by unconfoundedness (Assumption~\ref{ass:UC}(b)) and positivity, the second by consistency, the third by exclusion restriction (Assumption~\ref{ass:ER}), and the fourth by consistency and algebra. By monotonicity, we know \[A(z = 1) - A(z = 0) = \mathds{1}(A(z = 1) > A(z = 0)) \in \{0,1\},\]
    so we have that
    \begin{align*}
        &\mu_1(\boldsymbol{X}) - \mu_0(\boldsymbol{X}) \\
        &= \mathbb{E}_P\left(Y(a = 1) - Y(a = 0) \mid A(z = 1) > A(z = 0), \boldsymbol{X} \right) P[A(z = 1) > A(z = 0) \mid \boldsymbol{X}] \\
        &= \mathbb{E}_P\left(Y(a = 1) - Y(a = 0) \mid A(z = 1) > A(z = 0), \boldsymbol{X} \right) \mathbb{E}\left(A(z = 1) - A(z = 0) \mid \boldsymbol{X}\right) \\
        &= \mathbb{E}_P\left(Y(a = 1) - Y(a = 0) \mid A(z = 1) > A(z = 0), \boldsymbol{X} \right) \left\{\lambda_1(\boldsymbol{X}) - \lambda_0(\boldsymbol{X})\right\},
    \end{align*}
    where the last line again follows by unconfoundedness and positivity. Rearranging, we conclude that the CLATE equals $\frac{\mu_1(\boldsymbol{X}) - \mu_0(\boldsymbol{X})}{\lambda_1(\boldsymbol{X}) - \lambda_0(\boldsymbol{X})} = \Psi_P(\boldsymbol{X})$, by Lemma~\ref{lemma:binary}.
\end{proof}

\begin{proof}[Proof of Proposition~\ref{prop:one-sided}]
    Observe that $\mu_0(\boldsymbol{X})$ satisfies
    \[\mathbb{E}_P(Y \mid \boldsymbol{X}, Z = 0) = \mathbb{E}(Y(z = 0) \mid \boldsymbol{X}) = \mathbb{E}(Y(A(z = 0)) \mid \boldsymbol{X}) = \mathbb{E}(Y(a = 0) \mid \boldsymbol{X}),\]
    where the first equality holds by consistency and unconfoundedness (Assumption~\ref{ass:UC}(b)), the second by consistency and exclusion restriction (Assumption~\ref{ass:ER}), and the third by one-sided non-compliance. Thus,
    \[\mathbb{E}_P(Y \mid \boldsymbol{X}) - \mu_0(\boldsymbol{X}) = \mathbb{E}(Y - Y(a = 0) \mid \boldsymbol{X}) = \mathbb{E}(Y - Y(a = 0) \mid \boldsymbol{X}, A = 1)P[A = 1 \mid \boldsymbol{X}],\]
    by consistency. Rearranging, using $\mathbb{E}_P(Y \mid \boldsymbol{X}) = \mu_0(\boldsymbol{X}) + \pi_1(\boldsymbol{X}) \{\mu_1(\boldsymbol{X}) - \mu_0(\boldsymbol{X})\}$, and $P[A = 1 \mid \boldsymbol{X}] = \lambda_0(\boldsymbol{X}) + \pi_1(\boldsymbol{X}) \{\lambda_1(\boldsymbol{X}) - \lambda_0(\boldsymbol{X})\}$, and using one-sided non-compliance (i.e., $\lambda_0(\boldsymbol{X}) \equiv 0$), we obtain
    \[\mathbb{E}(Y(a = 1) - Y(a = 0) \mid \boldsymbol{X}, A = 1) = \frac{\mu_1(\boldsymbol{X}) - \mu_0(\boldsymbol{X})}{\lambda_1(\boldsymbol{X}) - \lambda_0(\boldsymbol{X})} = \Psi_P(\boldsymbol{X}),\]
    as claimed.
\end{proof}

\begin{proof}[Proof of Proposition~\ref{prop:stochastic-ident}]
    Observe that, for $z \in \{0,1\}$,
    \begin{align*}
        \mathbb{E}_P(AY \mid \boldsymbol{X}, Z = z)
        &= \mathbb{E}(A Y(a = 1) \mid \boldsymbol{X}, Z = z) \\
        &= \mathbb{E}(\mathbb{E}(AY(a = 1) \mid \boldsymbol{X}, U, Z = z) \mid \boldsymbol{X}, Z = z) \\
        &= \mathbb{E}(P[A = 1 \mid \boldsymbol{X}, U, Z = z]\mathbb{E}(Y(a = 1) \mid \boldsymbol{X}, U, Z = z) \mid \boldsymbol{X}, Z = z) \\
        &= \mathbb{E}(P[A = 1 \mid \boldsymbol{X}, U, Z = z]\mathbb{E}(Y(a = 1) \mid \boldsymbol{X}, U) \mid \boldsymbol{X}) \\
        &= \mathbb{E}(P[A = 1 \mid \boldsymbol{X}, U, Z = z]Y(a = 1) \mid \boldsymbol{X}).
    \end{align*}
    In the first equality we use consistency, in the second we use the tower law, in the third we use $A \independent Y(a) \mid \boldsymbol{X}, U, Z$, in the fourth we use unconfoundedness of $Z$, and in the fifth we use the tower law. Subtracting the above equality for $z = 1$ and $z = 0$, we obtain
    \[\mathbb{E}_P(AY \mid \boldsymbol{X}, Z =1) - \mathbb{E}_P(AY \mid \boldsymbol{X}, Z =0) = \mathbb{E}(w(\boldsymbol{X}, U) Y(a = 1) \mid \boldsymbol{X}).\]
    In the same way, we can show $\lambda_1(\boldsymbol{X}) - \lambda_0(\boldsymbol{X}) = \mathbb{E}(w(\boldsymbol{X}, U) \mid \boldsymbol{X})$ and
    \[\mathbb{E}_P((1 - A)Y \mid \boldsymbol{X}, Z =1) - \mathbb{E}_P((1 - A)Y \mid \boldsymbol{X}, Z =0) = -\mathbb{E}(w(\boldsymbol{X}, U) Y(a = 0) \mid \boldsymbol{X}).\]
    Adding these equalities gives $\mu_1(\boldsymbol{X}) - \mu_0(\boldsymbol{X}) = \mathbb{E}(\{Y(a = 1) - Y(a = 0)\}w(\boldsymbol{X}, U) \mid \boldsymbol{X})$, and dividing by $\lambda_1(\boldsymbol{X}) - \lambda_0(\boldsymbol{X}) = \mathbb{E}(w(\boldsymbol{X}, U) \mid \boldsymbol{X})$ yields~\eqref{eq:stochastic-ident}.
\end{proof}

\section{Proofs of Results in Section~\ref{sec:est-marginal}}

\begin{proof}[Proof of Lemma~\ref{lemma:VM-marginal}]
For brevity, we omit input $\boldsymbol{X}$ in the various nuisance functions. Beginning with the ATE parameter, observe that
\begin{align*}
    &\xi(\overline{P}) - \xi(P) + \mathbb{E}_P\left(\dot{\xi}(O;\overline{P})\right)
    \\
    &= \mathbb{E}_P\left(\frac{2Z - 1}{\overline{\delta}
\overline{\pi}_Z}\left\{Y - \overline{\mu}_Z -
\overline{\Psi}(A - \overline{\lambda}_Z)\right\} +
\overline{\Psi} - \Psi_P\right) \\
&= \mathbb{E}_P\left(\frac{1}{\overline{\delta}} \left\{\frac{\pi_1}{\overline{\pi}_1}\left(\mu_1 - \overline{\mu}_1 - \frac{\overline{\gamma}}{\overline{\delta}}(\lambda_1 - \overline{\lambda}_1)\right) - \frac{1 - \pi_1}{1 - \overline{\pi}_1}\left(\mu_0 - \overline{\mu}_0 - \frac{\overline{\gamma}}{\overline{\delta}}(\lambda_0 - \overline{\lambda}_0)\right)\right\} + \frac{\overline{\gamma}}{\overline{\delta}} - \frac{\gamma}{\delta}\right),
\end{align*}
where we conditioned on $(\boldsymbol{X}, Z)$ in the second equality. Since 
\begin{equation}\label{eq:bias-calc}
    \frac{\overline{\gamma}}{\overline{\delta}} - \frac{\gamma}{\delta} = \frac{1}{\overline{\delta}}\left\{(\overline{\delta} - \delta)\left(\frac{\overline{\gamma}}{\overline{\delta}} - \frac{\gamma}{\delta}\right) +  (\overline{\gamma} - \gamma) - \frac{\overline{\gamma}}{\overline{\delta}}(\overline{\delta} - \delta)\right\},
\end{equation}
we can combine terms to conclude that 
\begin{align*}
    &\xi(\overline{P}) - \xi(P) + \mathbb{E}_P\left(\dot{\xi}(O;\overline{P})\right)
    \\
    &= \mathbb{E}_P\bigg(\frac{1}{\overline{\delta}} \bigg\{\left(1 - \frac{\pi_1}{\overline{\pi}_1}\right)\left(\overline{\mu}_1 - \mu_1 - \frac{\overline{\gamma}}{\overline{\delta}}(\overline{\lambda}_1 - \lambda_1)\right) - \left(1 - \frac{1 - \pi_1}{1 - \overline{\pi}_1}\right)\left(\overline{\mu}_0 - \mu_0 - \frac{\overline{\gamma}}{\overline{\delta}}(\overline{\lambda}_0 - \lambda_0)\right) \\
    & \quad \quad \quad \quad \quad + (\overline{\delta} - \delta)\left(\frac{\overline{\gamma}}{\overline{\delta}} - \frac{\gamma}{\delta}\right)\bigg\}\bigg) \\
    &= \mathbb{E}_P\bigg(\frac{1}{\overline{\delta}}\bigg\{
(\overline{\pi}_1 - \pi_1)\bigg\{\frac{\overline{\mu}_1 - \mu_1 -
\frac{\overline{\gamma}}{\overline{\delta}}(\overline{\lambda}_1 - \lambda_1)}
{\overline{\pi}_1} +
\frac{\overline{\mu}_0 - \mu_0 -
\frac{\overline{\gamma}}{\overline{\delta}}(\overline{\lambda}_0 - \lambda_0)}
{1 - \overline{\pi}_1}\bigg\} \\
& \quad \quad \quad \quad \quad + (\overline{\delta} - \delta)
\left(\frac{\overline{\gamma}}{\overline{\delta}} -
\frac{\gamma}{\delta}\right)
\bigg\}\bigg),
\end{align*}
as claimed. Next, for the ATT parameter, observe that
\begin{align*}
    &\zeta(\overline{P}) - \zeta(P) + \mathbb{E}_P\left(\dot{\zeta}(O;\overline{P})\right)
    \\
    &= \left(\zeta(\overline{P}) - \zeta(P)\right)\left(1 - \frac{P[A = 1]}{\overline{P}[A = 1]}\right) \\
    & \quad \quad + \mathbb{E}_P\bigg(\frac{\overline{\rho}}{\overline{P}[A = 1]}\cdot
\frac{2Z - 1}{\overline{\delta}\overline{\pi}_Z}\left\{Y -
\overline{\mu}_Z - \overline{\Psi}(A -
\overline{\lambda}_Z)\right\} + \frac{A}{\overline{P}[A =
1]}\left(\overline{\Psi} - \Psi_P\right)
    \bigg) \\
&= \left(\zeta(\overline{P}) - \zeta(P)\right)\left(1 - \frac{P[A = 1]}{\overline{P}[A = 1]}\right) \\
    & \quad \quad + \frac{1}{\overline{P}[A = 1]}\mathbb{E}_P\bigg(\frac{\overline{\rho}}{\overline{\delta}}\left\{
\frac{\pi_1}{\overline{\pi}_1}\left(\mu_1 -
\overline{\mu}_1 - \frac{\overline{\gamma}}{\overline{\delta}}(\lambda_1 -
\overline{\lambda}_1)\right) - \frac{1 - \pi_1}{1 - \overline{\pi}_1}\left(\mu_0 -
\overline{\mu}_0 - \frac{\overline{\gamma}}{\overline{\delta}}(\lambda_0 -
\overline{\lambda}_0)\right)\right\} \\
& \quad \quad \quad \quad \quad + \rho \left(\frac{\overline{\gamma}}{\overline{\delta}} - \frac{\gamma}{\delta}\right)
    \bigg)
\end{align*}
By~\eqref{eq:bias-calc}, we have
\[\rho\left(\frac{\overline{\gamma}}{\overline{\delta}} - \frac{\gamma}{\delta}\right) = (\rho - \overline{\rho})\left(\frac{\overline{\gamma}}{\overline{\delta}} - \frac{\gamma}{\delta}\right) + \frac{\overline{\rho}}{\overline{\delta}}\left\{(\overline{\delta} - \delta)\left(\frac{\overline{\gamma}}{\overline{\delta}} - \frac{\gamma}{\delta}\right) +  (\overline{\gamma} - \gamma) - \frac{\overline{\gamma}}{\overline{\delta}}(\overline{\delta} - \delta)\right\},\]
so combining terms we obtain
\begin{align*}
&\zeta(\overline{P}) - \zeta(P) + \mathbb{E}_P\left(\dot{\zeta}(O;\overline{P})\right)
    \\
    &= \frac{1}{\overline{P}[A=1]}\mathbb{E}_P\bigg(\frac{\overline{\rho}}{\overline{\delta}}\bigg\{
(\overline{\pi}_1 - \pi_1)\bigg\{\frac{\overline{\mu}_1 - \mu_1 -
\frac{\overline{\gamma}}{\overline{\delta}}(\overline{\lambda}_1 - \lambda_1)}
{\overline{\pi}_1} +
\frac{\overline{\mu}_0 - \mu_0 -
\frac{\overline{\gamma}}{\overline{\delta}}(\overline{\lambda}_0 - \lambda_0)}
{1 - \overline{\pi}_1}\bigg\}\bigg\} \\
& \quad \quad \quad \quad \quad + \bigg(\frac{\overline{\gamma}}{\overline{\delta}}
- \frac{\gamma}{\delta}\bigg)\bigg[(\rho - \overline{\rho})+
\frac{\overline{\rho}}{\overline{\delta}}(\overline{\delta} - \delta)\bigg] \bigg) \\
& \quad \quad + \left(\zeta(\overline{P}) - \zeta(P)\right)\left(1 - \frac{P[A = 1]}{\overline{P}[A = 1]}\right),
\end{align*}
as claimed. Finally, for the LATE/SIW-WATE parameter, observe that
\begin{align*}
    &\chi(\overline{P}) - \chi(P) + \mathbb{E}_P\left(\dot{\chi}(O;\overline{P})\right)
    \\
    &= \left(\chi(\overline{P}) - \chi(P)\right)\left(1 - \frac{\Delta(P)}{\Delta(\overline{P})}\right) + \frac{\chi(\overline{P})}{\Delta(\overline{P})}\mathbb{E}_P(\lambda_1 - \lambda_0) - \frac{1}{\Delta(\overline{P})}\mathbb{E}_P(\mu_1 - \mu_0)\\
    & \quad \quad + \mathbb{E}_P\bigg(\frac{2Z - 1}{\Delta(\overline{P})
\overline{\pi}_Z}\left\{Y - \overline{\mu}_Z - \chi(\overline{P})(A -
\overline{\lambda}_Z)\right\} +
\frac{1}{\Delta(\overline{P})}\left(\overline{\gamma} -
\chi(\overline{P})\overline{\delta}\right)\bigg) \\
&= \left(\chi(\overline{P}) - \chi(P)\right)\left(1 - \frac{\Delta(P)}{\Delta(\overline{P})}\right) + \frac{\chi(\overline{P})}{\Delta(\overline{P})}\mathbb{E}_P(\lambda_1 - \lambda_0) - \frac{1}{\Delta(\overline{P})}\mathbb{E}_P(\mu_1 - \mu_0)\\
    & \quad \quad + 
    \frac{1}{\Delta(\overline{P})}\mathbb{E}_P\bigg(\frac{\pi_1}{\overline{\pi}_1}\left\{\mu_1 - \overline{\mu}_1 - \chi(\overline{P})(\lambda_1 -
\overline{\lambda}_1)\right\} - \frac{1 - \pi_1}{1 - \overline{\pi}_1}\left\{\mu_0 - \overline{\mu}_0 - \chi(\overline{P})(\lambda_0 -
\overline{\lambda}_0)\right\}\\
& \quad \quad \quad \quad \quad \quad \quad \quad \quad + \overline{\gamma} -
\chi(\overline{P})\overline{\delta}\bigg) \\
&= \left(\chi(\overline{P}) - \chi(P)\right)\left(1 - \frac{\Delta(P)}{\Delta(\overline{P})}\right)
- \frac{\chi(\overline{P})}{\Delta(\overline{P})}
\mathbb{E}_P\left((\overline{\pi}_1 - \pi_1) \left\{\frac{\overline{\lambda}_1 - \lambda_1}{\overline{\pi}_1} + \frac{\overline{\lambda}_0 - \lambda_0}{1 - \overline{\pi}_1}\right\}\right) \\
& \quad \quad \quad \quad \quad + \frac{1}{\Delta(\overline{P})} \mathbb{E}_P\left((\overline{\pi}_1 - \pi_1) \left\{\frac{\overline{\mu}_1 - \mu_1}{\overline{\pi}_1} + \frac{\overline{\mu}_0 - \mu_0}{1 - \overline{\pi}_1}\right\}\right),
\end{align*}
as claimed.
\end{proof}

\begin{proof}[Proof of Corollary~\ref{cor:IF-marginal}]
This follows from Lemma \ref{lemma:VM-marginal} and a direct application of Lemma 2 in \citet{kennedy2023}
\end{proof}

\begin{proof}[Proof of Theorem~\ref{thm:est-marginal}]
    For brevity, we write $P(f) = \mathbb{E}_P(f(O))$ for the mean of any function $f$ under $P$. Starting with the ATE parameter $\xi(P)$, observe that, by the definition of the one-step estimator $\widehat{\xi} = \xi(\widehat{P}) + \mathbb{P}_n\left(\dot{\xi}(O; \widehat{P})\right)$, we can decompose error as follows:
    \[\widehat{\xi} - \xi(P) = \mathbb{P}_n\left(\dot{\xi}(O;P)\right) +(\mathbb{P}_n - P)\left(\dot{\xi}(O;\widehat{P}) - \dot{\xi}(O;P)\right) + R_{\xi}(\widehat{P}, P),\]
    as by definition, $R_{\xi}(\widehat{P}, P) = \xi(\widehat{P}) - \xi(P) + P(\dot{\xi}(O;\widehat{P}))$, and $P(\dot{\xi}(O;P)) = 0$. Since $\widehat{P}$ is trained on data independent of $(O_1, \ldots, O_n)$, and by the assumption that $\lVert \xi(\, \cdot \, ; \widehat{P}) - \xi(\, \cdot \, ; P)\rVert = o_P(1)$, it follows by Lemma 2 of~\citet{kennedy2020b} that $(\mathbb{P}_n - P)\left(\dot{\xi}(O;\widehat{P}) - \dot{\xi}(O;P)\right) = o_P(n^{-1/2})$, which proves the claim for $\widehat{\xi}$. The argument for the one-step estimator $\widehat{\zeta}$ is identical.
    
    For $\widehat{\chi}$, we follow the argument in Theorem 3 of~\citet{kennedy2020a}. Note that $\mathbb{E}_P(\phi_a(O;P)) = \Delta(P) \neq 0$ by assumption. Observe that
    \begin{align*}
        &\widehat{\chi} - \chi(P) \\
        &= \frac{\mathbb{P}_n(\phi_y(O; \widehat{P}))}{\mathbb{P}_n(\phi_a(O; \widehat{P}))} - \frac{P(\phi_y(O; P))}{P(\phi_a(O; P))} \\
        &= \frac{1}{\mathbb{P}_n(\phi_a(O; \widehat{P}))}\left\{\mathbb{P}_n(\phi_y(O; \widehat{P})) - P(\phi_y(O; P)) - \chi(P)\left(\mathbb{P}_n(\phi_a(O; \widehat{P})) - P(\phi_a(O; P))\right)\right\}.
    \end{align*}
    Following the calculations in the  proof of Lemma~\ref{lemma:VM-marginal}, we can decompose
    \begin{align*}
        &\mathbb{P}_n(\phi_b(O; \widehat{P})) - P(\phi_b(O; P))  \\
        &= (\mathbb{P}_n - P)\left(\phi_b(O;P)\right) +(\mathbb{P}_n - P)\left(\phi_b(O;\widehat{P}) - \phi_b(O;P)\right) + R_{b}(\widehat{P}, P),
    \end{align*}
    for $b \in \{a, y\}$. The second terms in these expansions are $o_P(n^{-1/2})$ by Lemma 2 of \citet{kennedy2020a}. Thus, noting that $\dot{\chi}(O; P) = \frac{1}{\Delta(P)}\left\{\phi_y(O; P) - \chi(P) \phi_a(O; P)\right\}$, and using the fact $\mathbb{P}_n(\phi_a(O; \widehat{P}))$ is bounded away from zero, it follows that
    \begin{align*}
        \widehat{\chi} - \chi(P)
        = \mathbb{P}_n\left(\dot{\chi}(O; P)\right) + O_P\left(R_y + R_a\right) + o_P\left(n^{-1/2}\right),
    \end{align*}
    since
    \begin{align*}
        &\left\{\frac{1}{\mathbb{P}_n(\phi_a(O; \widehat{P}))} - \frac{1}{\Delta(P)}\right\}(\mathbb{P}_n - P)\left(\phi_y(O;P) - \chi(P) \phi_a(O;P)\right) \\
        &= \frac{\Delta(P) - \mathbb{P}_n(\phi_a(O; \widehat{P}))}{\mathbb{P}_n(\phi_a(O; \widehat{P}))\Delta(P)} \cdot O_P(n^{-1/2}) = o_P(1)O_P(n^{-1/2}) = o_P(n^{-1/2}),
    \end{align*}
    where the first equality follows by the central limit theorem, and the second equality as $\mathbb{P}_n(\phi_a(O; \widehat{P}))$ is bounded away from zero and
    \begin{align*}
        &\mathbb{P}_n(\phi_a(O; \widehat{P})) - P(\phi_a(O;P)) \\
        & = (\mathbb{P}_n - P)(\phi_a(O; P)) + (\mathbb{P}_n - P)(\phi_a(O; \widehat{P}) - \phi_a(O; P)) + P(\phi_a(O; \widehat{P}) - \phi_a(O; P)) \\
        &= O_P(n^{-1/2}) + O_P(n^{-1/2}) + o_P(1) \\
        &= o_P(1),
    \end{align*}
    again using the central limit theorem, Lemma 2 of \citet{kennedy2020a}, and noting that by assumption, $\left|P(\phi_a(O; \widehat{P}) - \phi_a(O; P))\right| \lesssim \lVert \phi_a(\, \cdot \, ; \widehat{P}) - \phi_a(\, \cdot \, ; P)\rVert = o_P(1)$.
\end{proof}

\section{Proofs of Results in Section~\ref{sec:est-project}}

\begin{proof}[Proof of Proposition~\ref{prop:IF-projection}]
    For each $\theta \in \{\xi, \zeta, \chi\}$, we implicitly assume that $m_\theta(V;\boldsymbol{\beta})$ is differentiable in $\boldsymbol{\beta}$, and that the minimizer in~\eqref{eq:projection} is unique. It follows that $\boldsymbol{\beta}_{\theta}(P)$ uniquely satisfies the estimating equations
    \begin{equation}\label{eq:projection-est-eq}
        0 = \mathbb{E}_{P}\left(M_{\theta}(V; \boldsymbol{\beta}_{\theta}(P)) w(V)\left\{\theta_P(V) -
          m_{\theta}(V; \boldsymbol{\beta}_{\theta}(P))\right\}\right),
    \end{equation}
    where $M_{\theta}(V; \boldsymbol{\beta}) = \nabla_{\boldsymbol{\beta}}\,
    m_{\theta}(V; \boldsymbol{\beta})$. We will use this characterizing equation to deduce the influence function of $\boldsymbol{\beta}_{\theta}(P)$. Recall from \citet{bkrw1993} and \citet{vandervaart2002} that the nonparametric influence function $\dot{\boldsymbol{\beta}}_{\theta}(O;P)$ is the unique mean-zero finite variance function satisfying pathwise differentiability:
    \[\left. \frac{d}{d\epsilon}\boldsymbol{\beta}_{\theta}(P_{\epsilon})\right|_{\epsilon = 0} = \mathbb{E}_P\left(\dot{\boldsymbol{\beta}}_{\theta}(O;P) u(O)\right),\]
    for any regular parametric submodel $\{P_{\epsilon}: \epsilon \in [0,1)\}$ such that $P_0 \equiv P$ with score function $u(O) = \left. \frac{d}{d\epsilon} \log{d P_{\epsilon}} \right|_{\epsilon = 0}$. Choosing such a regular parametric submodel and differentiating~\eqref{eq:projection-est-eq},
    \begin{align*}
        0 
        &= \left. \frac{d}{d\epsilon} \mathbb{E}_{P_{\epsilon}}\left(M_{\theta}(V; \boldsymbol{\beta}_{\theta}(P)) w(V)\left\{\theta_{P}(V) -
          m_{\theta}(V; \boldsymbol{\beta}_{\theta}(P))\right\}\right) \right|_{\epsilon = 0} \\
        & \quad \quad + \left. \frac{d}{d\epsilon} \mathbb{E}_{P}\left(M_{\theta}(V; \boldsymbol{\beta}_{\theta}(P_{\epsilon})) w(V)\left\{\theta_{P}(V) -
          m_{\theta}(V; \boldsymbol{\beta}_{\theta}(P))\right\}\right) \right|_{\epsilon = 0} \\
        & \quad \quad + \left. \frac{d}{d\epsilon} \mathbb{E}_{P}\left(M_{\theta}(V; \boldsymbol{\beta}_{\theta}(P)) w(V)\left\{\theta_{P_{\epsilon}}(V) -
          m_{\theta}(V; \boldsymbol{\beta}_{\theta}(P))\right\}\right) \right|_{\epsilon = 0} \\
        & \quad \quad + \left. \frac{d}{d\epsilon} \mathbb{E}_{P}\left(M_{\theta}(V; \boldsymbol{\beta}_{\theta}(P)) w(V)\left\{\theta_{P}(V) -
          m_{\theta}(V; \boldsymbol{\beta}_{\theta}(P_{\epsilon}))\right\}\right) \right|_{\epsilon = 0},
    \end{align*}
    invoking the total derivative. The first summand above is simply
    \[\mathbb{E}_{P}\left(M_{\theta}(V; \boldsymbol{\beta}_{\theta}(P)) w(V)\left\{\theta_{P}(V) -
          m_{\theta}(V; \boldsymbol{\beta}_{\theta}(P))\right\}u(O)\right),\]
    by interchanging integral and derivative, noting that $\left. \frac{d}{d\epsilon} \log{d P_{\epsilon}} \right|_{\epsilon = 0} = \left. \frac{d}{d\epsilon}dP_{\epsilon}\right|_{\epsilon = 0}/dP$. The second and fourth summands are
    \[\mathbb{E}_{P}\left(\Lambda_{\theta}(V; \boldsymbol{\beta}_{\theta}(P)) w(V)\left\{\theta_{P}(V) -
          m_{\theta}(V; \boldsymbol{\beta}_{\theta}(P))\right\}\right) \left. \frac{d}{d\epsilon}\boldsymbol{\beta}_{\theta}(P_{\epsilon})\right|_{\epsilon = 0},\]
    \[-\mathbb{E}_{P}\left(M_{\theta}^{\otimes 2}(V; \boldsymbol{\beta}_{\theta}(P)) w(V)\right) \left. \frac{d}{d\epsilon}\boldsymbol{\beta}_{\theta}(P_{\epsilon})\right|_{\epsilon = 0},\]
    respectively, by the chain rule, where
  $\Lambda_{\theta}(V; \boldsymbol{\beta}) =
  \nabla_{\boldsymbol{\beta}}^2 m_{ \theta}(V;
  \boldsymbol{\beta})$ and $\boldsymbol{u}^{\otimes 2} = \boldsymbol{u} \boldsymbol{u}^T$ for any $\boldsymbol{u} \in \mathbb{R}^q$. Rewriting the third summand and rearranging the above equality, we therefore obtain
  \begin{align*}
      & \left. \frac{d}{d\epsilon}\boldsymbol{\beta}_{\theta}(P_{\epsilon})\right|_{\epsilon = 0} \\
      &= -V_{\theta}^{-1}\mathbb{E}_{P}\bigg(M_{\theta}(V; \boldsymbol{\beta}_{\theta}(P)) w(V)\left[\left\{\theta_{P}(V) -
          m_{\theta}(V; \boldsymbol{\beta}_{\theta}(P))\right\}u(O) + \left.\frac{d}{d\epsilon} \theta_{P_{\epsilon}}(V) \right|_{\epsilon=0}\right]\bigg),
  \end{align*}
  where $V_{\theta} = \mathbb{E}_P\left(w(V)\left\{\Lambda_{\theta}(V;
          \boldsymbol{\beta}_{\theta}(P))\left[\theta_P(V) - m_{
              \theta}(V; \boldsymbol{\beta}_{\theta}(P))\right] -
          M_{\theta}^{\otimes 2}(V;
          \boldsymbol{\beta}_{\theta}(P))\right\}\right)$; note that this normalizing matrix equals $V_{\theta}(\boldsymbol{\beta}_{\theta}(P), \eta_{\theta})$ defined in Theorem~\ref{thm:est-projection}.
    
    Now, it remains to compute $\left.\frac{d}{d\epsilon} \theta_{P_{\epsilon}}(V) \right|_{\epsilon=0}$ for each $\theta \in \{\xi, \zeta, \chi\}$ to obtain the formulae given in the main paper. First,
    \begin{align*}
        \left.\frac{d}{d\epsilon} \xi_{P_{\epsilon}}(V) \right|_{\epsilon=0}
        &= \left.\frac{d}{d\epsilon} \mathbb{E}_{P_{\epsilon}}\left(\Psi_{\epsilon}(\boldsymbol{X}) \mid V\right)\right|_{\epsilon=0} \\
        &= \mathbb{E}_P\left(\left\{\Psi_P(\boldsymbol{X}) - \xi_P(V)\right\}u_{\boldsymbol{X} \mid V} \mid V\right) + \mathbb{E}_P\left(\left.\frac{d}{d\epsilon} \Psi_{\epsilon}(\boldsymbol{X}) \right|_{\epsilon = 0} \mid V\right),
    \end{align*}
    where $u_{B \mid C}$ represents the conditional score function for $B$ given $C$, for arbitrary $B$ and $C$. Since $\Psi_\epsilon = \gamma_{\epsilon}/\delta_{\epsilon}$, the quotient rule yields
    \begin{align*}
        \left.\frac{d}{d\epsilon} \Psi_{\epsilon}(\boldsymbol{X}) \right|_{\epsilon = 0} 
        &= \frac{1}{\delta(\boldsymbol{X})}\left( \left.\frac{d}{d\epsilon} \gamma_{\epsilon}(\boldsymbol{X}) \right|_{\epsilon = 0} - \Psi_P(\boldsymbol{X}) \left.\frac{d}{d\epsilon} \delta_{\epsilon}(\boldsymbol{X}) \right|_{\epsilon = 0}\right).
    \end{align*}
    Continuing the calculation, we observe that
    \begin{align*}
        \left.\frac{d}{d\epsilon} \mu_{z,\epsilon}(\boldsymbol{X}) \right|_{\epsilon = 0}
        &= \mathbb{E}_P(\left\{Y - \mu_z(\boldsymbol{X})\right\}u_{Y \mid \boldsymbol{X}, Z = z} \mid \boldsymbol{X}, Z = z) \\
        &= \mathbb{E}_P\left(\frac{\mathds{1}(Z = z)}{\pi_z(\boldsymbol{X})}\left\{Y - \mu_z(\boldsymbol{X})\right\}u_{Y \mid \boldsymbol{X}, Z} \, \middle| \,\boldsymbol{X}\right),
    \end{align*}
    and similarly
    \[\left.\frac{d}{d\epsilon} \lambda_{z,\epsilon}(\boldsymbol{X}) \right|_{\epsilon = 0}
        = \mathbb{E}_P\left(\frac{\mathds{1}(Z = z)}{\pi_z(\boldsymbol{X})}\left\{A - \lambda_z(\boldsymbol{X})\right\}u_{A \mid \boldsymbol{X}, Z} \, \middle| \,\boldsymbol{X}\right).\]
    Repeatedly using the tower law, as well as properties of (conditional) score functions (i.e., they have conditional mean zero and $u_{B \mid C} + u_C = u_{B, C}$ for any $B$ and $C$), we obtain
    \begin{align*}
      & \left. \frac{d}{d\epsilon}\boldsymbol{\beta}_{\xi}(P_{\epsilon})\right|_{\epsilon = 0} \\
      &= -V_{\xi}^{-1}\mathbb{E}_{P}\bigg(M_{\xi}(V; \boldsymbol{\beta}_{\xi}(P)) w(V)\bigg[\left\{\xi_{P}(V) -
          m_{\xi}(V; \boldsymbol{\beta}_{\xi}(P))\right\}u(O) \\
    & \quad \quad \quad
    + \left\{\Psi_P(\boldsymbol{X}) - \xi_P(V) + \frac{2Z-1}{\delta(\boldsymbol{X})\pi_Z(\boldsymbol{X})}\left\{Y - \mu_Z(\boldsymbol{X}) -\Psi_P(\boldsymbol{X})(A - \lambda_Z(\boldsymbol{X}))\right\}\right\} u(O)\bigg]\bigg), \\
    &= -V_{\xi}^{-1}\mathbb{E}_{P}\bigg(M_{\xi}(V; \boldsymbol{\beta}_{\xi}(P)) w(V)\bigg[\Psi_P(\boldsymbol{X}) -
          m_{\xi}(V; \boldsymbol{\beta}_{\xi}(P)) \\
    & \quad \quad \quad \quad \quad \quad \quad 
     + \frac{2Z-1}{\delta(\boldsymbol{X})\pi_Z(\boldsymbol{X})}\left\{Y - \mu_Z(\boldsymbol{X}) -\Psi_P(\boldsymbol{X})(A - \lambda_Z(\boldsymbol{X}))\right\} \bigg]u(O)\bigg).
  \end{align*}
  By definition, we conclude that
  \begin{align*}
      \dot{\boldsymbol{\beta}}_{\xi}(O; P)
      &= -V_{\xi}^{-1} M_{\xi}(V; \boldsymbol{\beta}_{\xi}(P)) w(V)\bigg[\Psi_P(\boldsymbol{X}) -
          m_{\xi}(V; \boldsymbol{\beta}_{\xi}(P)) \\
    & \quad \quad \quad \quad
     + \frac{2Z-1}{\delta(\boldsymbol{X})\pi_Z(\boldsymbol{X})}\left\{Y - \mu_Z(\boldsymbol{X}) -\Psi_P(\boldsymbol{X})(A - \lambda_Z(\boldsymbol{X}))\right\} \bigg],
  \end{align*}
  as claimed. For $\zeta_P(V)$, observe that
  \begin{align*}
        &\left.\frac{d}{d\epsilon} \zeta_{P_{\epsilon}}(V) \right|_{\epsilon=0} \\
        &= \left.\frac{d}{d\epsilon} \mathbb{E}_{P_{\epsilon}}\left(\Psi_{\epsilon}(\boldsymbol{X}) \mid V, A = 1\right)\right|_{\epsilon=0} \\
        &= \mathbb{E}_P\left(\left\{\Psi_P(\boldsymbol{X}) - \zeta_P(V)\right\}u_{\boldsymbol{X} \mid V, A = 1} \mid V, A = 1\right) + \mathbb{E}_P\left(\left.\frac{d}{d\epsilon} \Psi_{\epsilon}(\boldsymbol{X}) \right|_{\epsilon = 0} \mid V, A = 1\right) \\
        &= \mathbb{E}_P\left(\frac{A}{P[A = 1 \mid V]}\left\{\Psi_P(\boldsymbol{X}) - \zeta_P(V)\right\}u_{\boldsymbol{X} \mid V, A} \mid V\right) \\
        & \quad \quad + \mathbb{E}_P\bigg(\frac{A}{P[A = 1 \mid V]}\frac{1}{\delta(\boldsymbol{X})}\mathbb{E}_P\bigg(\frac{2Z-1}{\pi_Z(\boldsymbol{X})}(Y - \mu_Z(\boldsymbol{X}))u_{Y \mid \boldsymbol{X}, Z} \\
        & \hspace{2.2in} - \Psi_P(\boldsymbol{X})(A - \lambda_Z(\boldsymbol{X}))u_{A\mid \boldsymbol{X}, Z} \, \bigg| \, \boldsymbol{X}\bigg)\mid V\bigg),
    \end{align*}
    using the previous calculations for $\left.\frac{d}{d\epsilon} \Psi_{\epsilon}(\boldsymbol{X}) \right|_{\epsilon = 0}$. Again using the tower law and properties of score functions, we conclude that
    \begin{align*}
      \dot{\boldsymbol{\beta}}_{\zeta}(O; P)
      &= -V_{\zeta}^{-1} M_{\zeta}(V; \boldsymbol{\beta}_{\zeta}(P)) w(V)\bigg[\zeta_P(V) -
          m_{\zeta}(V; \boldsymbol{\beta}_{\zeta}(P)) \\
    & \quad \quad \quad \quad
     + \frac{1}{P[A = 1 \mid V]}\bigg(A \left\{\Psi_P(\boldsymbol{X}) - \zeta_P(V)\right\} \\
    & \quad \quad \quad \quad \quad \quad \quad + \frac{\rho(\boldsymbol{X})}{\delta(\boldsymbol{X})}\frac{2Z-1}{\pi_Z(\boldsymbol{X})}\left\{Y - \mu_Z(\boldsymbol{X}) -\Psi_P(\boldsymbol{X})(A - \lambda_Z(\boldsymbol{X}))\right\} \bigg)\bigg].
  \end{align*}
  Lastly, for $\chi_P(V) = \frac{\mathbb{E}_P(\gamma(\boldsymbol{X}) \mid V)}{\mathbb{E}_P(\delta(\boldsymbol{X}) \mid V)}$, applying the quotient rule gives
  \begin{align*}
        \left.\frac{d}{d\epsilon} \chi_{P_{\epsilon}}(V) \right|_{\epsilon=0}
        &= \frac{1}{\mathbb{E}_P(\delta(\boldsymbol{X}) \mid V)}\left( \left.\frac{d}{d\epsilon} \mathbb{E}_{P_{\epsilon}}(\gamma_{\epsilon}(\boldsymbol{X}) \mid V) \right|_{\epsilon = 0} - \chi_P(V) \left.\frac{d}{d\epsilon} \mathbb{E}_{P_{\epsilon}}(\delta_{\epsilon}(\boldsymbol{X}) \mid V) \right|_{\epsilon = 0}\right).
    \end{align*}
    Noting that
    \begin{align*}
        \left.\frac{d}{d\epsilon} \mathbb{E}_{P_{\epsilon}}(\gamma_{\epsilon}(\boldsymbol{X}) \mid V)\right|_{\epsilon = 0} 
        &= \mathbb{E}_P\left(\left\{\gamma(\boldsymbol{X}) -\mathbb{E}_P(\gamma(\boldsymbol{X}) \mid V)\right\}u_{\boldsymbol{X} \mid V} \mid V\right) \\
        & \quad \quad \quad \quad + \mathbb{E}_P\left(\mathbb{E}_P\left\{\frac{2Z-1}{\pi_Z(\boldsymbol{X})}(Y - \mu_Z(\boldsymbol{X}))u_{Y \mid \boldsymbol{X}, Z} \, \middle| \, \boldsymbol{X}\right\} \mid V\right)
    \end{align*}
    and
    \begin{align*}
        \left.\frac{d}{d\epsilon} \mathbb{E}_{P_{\epsilon}}(\delta_{\epsilon}(\boldsymbol{X}) \mid V)\right|_{\epsilon = 0}
        &= \mathbb{E}_P\left(\left\{\delta(\boldsymbol{X}) -\mathbb{E}_P(\delta(\boldsymbol{X}) \mid V)\right\}u_{\boldsymbol{X} \mid V} \mid V\right) \\
        & \quad \quad \quad \quad + \mathbb{E}_P\left(\mathbb{E}_P\left\{\frac{2Z-1}{\pi_Z(\boldsymbol{X})}(A - \lambda_Z(\boldsymbol{X}))u_{A \mid \boldsymbol{X}, Z} \, \middle| \, \boldsymbol{X}\right\} \mid V\right)
    \end{align*}
    Thus, using the tower law and properties of scores, and simplifying, we obtain
    \begin{align*}
      \dot{\boldsymbol{\beta}}_{\chi}(O; P)
      &= -V_{\chi}^{-1} M_{\chi}(V; \boldsymbol{\beta}_{\chi}(P)) w(V)\bigg[\chi_P(V) -
          m_{\chi}(V; \boldsymbol{\beta}_{\chi}(P)) \\
    & \quad \quad \quad \quad
     + \frac{1}{\mathbb{E}_P(\delta(\boldsymbol{X}) \mid V)}\bigg(\gamma(\boldsymbol{X}) - \chi_P(V)\delta(\boldsymbol{X}) \\
    & \quad \quad \quad \quad \quad \quad \quad + \frac{2Z-1}{\pi_Z(\boldsymbol{X})}\left\{Y - \mu_Z(\boldsymbol{X}) -\chi_P(V)(A - \lambda_Z(\boldsymbol{X}))\right\} \bigg)\bigg],
  \end{align*}
  as claimed.
\end{proof}

\begin{proof}[Proof of Theorem~\ref{thm:est-projection}]
    This result is a direct application of Lemma 3 in~\citet{kennedy2023}, a ``master lemma'' giving rates of convergence for solutions to sample-split estimating equations involving estimated nuisance functions. We need only justify the stated formulas for the bias terms $R_{n, \theta} = \mathbb{E}_P\left(\phi_{\theta}(O;
      \boldsymbol{\beta}_{\theta}(P), \widehat{\eta}_{\theta}) - \phi_{\theta}(O;
      \boldsymbol{\beta}_{\theta}(P), \eta_{\theta})\right)$, for $\theta \in \{\xi, \zeta, \chi\}$, but this follows similar arguments to the second-order remainder calculations in the proof of Lemma~\ref{lemma:VM-marginal}.
\end{proof}

\section{IV Unconfoundedness Assumption Revisited}

In Remark~\ref{remark:UC} in Section~\ref{sec:notation}, we noted that not all listed unconfoundedness assumptions need be invoked for every identification result. Indeed, Propositions~\ref{prop:homogeneity-ident}, \ref{prop:nonconstant-ident}, \ref{prop:SMM-ident}, and \ref{prop:hernan} rely only on unconfoundedness through Assumption~\ref{ass:UC}(a), i.e., $Z \independent Y(z, a) \mid \boldsymbol{X}$. In this section, we briefly describe a scenario where this latter condition holds, yet other listed unconfoundedness assumptions fail.

Consider the putative data generating scenario illustrated in Figure~\ref{fig:IV-alt}. In addition to unmeasured $A$--$Y$ confounders (here written as $U_1$), we included additional $Z$--$A$ confounders, $U_2$. One can rigorously assess the counterfactual independencies implied by this DAG by associating it with a formal causal model; here, we will use the finest fully randomized causally interpretable structured tree graph model \citep{robins2010}, for which many counterfactual independencies can be probed with simple graphical criteria \citep{richardson2013}. In this DAG, due to the independence of $U_1$ and $U_2$ given $\boldsymbol{X}$, Assumption~\ref{ass:UC}(a) holds, i.e., $\boldsymbol{X}$ is sufficient to control for $Z$--$Y$ confounding. On the other hand, Assumption~\ref{ass:UC}(b) fails, since $Z \not \independent A(z) \mid \boldsymbol{X}$ due to unmeasured confounding by $U_2$. Note, in this case that Assumption~\ref{ass:UC}(c) and \ref{ass:UC}(d) hold here by choosing $U \equiv U_1$.

\begin{figure}[ht]
\centering




\large{\begin{tikzpicture}[%
        ->,
        >=stealth,
        node distance=0.5cm,
        pil/.style={
          ->,
          thick,
          shorten =2pt,},
        regnode/.style={circle, draw=black, fill=none, thick, minimum size=8mm},
        bluenode/.style={circle, draw=blue, fill=none, thick, minimum size=8mm},
        boxednode/.style={rectangle, draw=black, fill=none, thick, minimum size=8mm},
        leftsplitnode/.style={semicircle, draw=black, fill=none, thick, minimum size=9mm, shape border rotate=90},
        rightsplitnode/.style={semicircle, draw=red, fill=none, thick, minimum size=9mm, shape border rotate=270},
        rednode/.style={circle, draw=red, fill=none, thick, minimum size=8mm},
        redbox/.style={rectangle, draw=red, fill=none, thick, minimum size=8mm},
        ]
        \node[regnode] (1) {$Z$};
        \node[regnode, right=of 1] (2) {$A$};
        \node[regnode, above right=of 2] (3) {$U_1$};
        \node[regnode, below right=of 3] (4) {$Y$};
        \node[regnode, left=of 1] (5) {$\boldsymbol{X}$};
        \node[regnode, below=of 5] (6) {$U_2$};
        \draw [->] (1) to (2);
        \draw [->] (2) to (4);
        \draw [->] (3) to (2);
        \draw [->] (3) to (4);
        \draw [->] (5) to (1);
        \draw [->] (5) to [out=330, in=210] (2);
        \draw [<->] (5) to [out=45, in=155] (3);
        \draw [<->] (5) to [out=225, in=135] (6);
        \draw [->] (5) to [out=315, in=225] (4);
        \draw [->] (6) to (1);
        \draw [->] (6) to [out = 0, in = 240] (2);
\end{tikzpicture}}

\caption{Non-standard instrumental variable DAG}
\label{fig:IV-alt}
\end{figure}
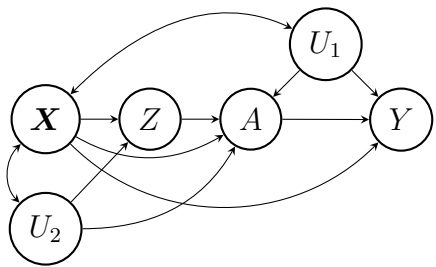


\end{appendices}

\end{document}